\documentclass[lettersize,journal]{IEEEtran}

\usepackage{amsmath,amsfonts,bm}

\def\eqref#1{equation~\ref{#1}}

\def\1{\bm{1}}

\DeclareMathAlphabet{\mathsfit}{\encodingdefault}{\sfdefault}{m}{sl}
\SetMathAlphabet{\mathsfit}{bold}{\encodingdefault}{\sfdefault}{bx}{n}

\usepackage{amsthm}

\newtheoremstyle{nonitalictheorem}
  {\topsep}         
  {\topsep}         
  {\normalfont}     
  {}                
  {\bfseries}       
  {.}               
  {.5em}            
  {}                

\usepackage{hyperref}
\usepackage{url}
\usepackage{amsmath}
\usepackage{amsfonts}
\usepackage{latexsym}
\usepackage{enumerate}
\usepackage{wasysym}
\usepackage{lipsum}
\usepackage{subfig}
\usepackage{graphicx}
\usepackage{float}
\usepackage{amssymb}
\usepackage{enumitem}
\usepackage{amsthm}
\newtheorem{theorem}{Theorem}
\newtheorem{assumption}{Assumption}

\newtheorem{Corollary}{Corollary}
\theoremstyle{nonitalictheorem}
\newtheorem{Remark}{Remark}
\usepackage{wrapfig,commath}
\usepackage{color}
\usepackage{algorithm}
\usepackage{algorithmic}
\usepackage{booktabs}
\usepackage{multirow}
\usepackage{array}
\allowdisplaybreaks[4]

\ifCLASSOPTIONcompsoc
  \usepackage[nocompress]{cite}
\else
  \usepackage{cite}
\fi
\ifCLASSINFOpdf
\else
\fi
\begin{document}
%
\title{Accelerating Hybrid Federated Learning Convergence under Partial Participation}

\author{Jieming Bian, Lei Wang, Kun Yang, Cong Shen,~\IEEEmembership{Senior Member,~IEEE,}
    Jie Xu,~\IEEEmembership{Senior Member,~IEEE}
\thanks{Jieming Bian, Lei Wang and Jie Xu are with the Department of Electrical and Computer Engineering, University of Miami, Coral Gables,
FL 33146, USA. Email: \{jxb1974, lxw725, jiexu\}@miami.edu. 
}
\thanks{Kun Yang and Cong Shen are with the Department of Electrical and Computer Engineering, University of Virginia, Charlottesville, VA 22904, USA. Email: \{ky9tc, cong\}@virginia.edu. 
}

}

\maketitle

\begin{abstract}
Over the past few years, Federated Learning (FL) has become a popular distributed machine learning paradigm. FL involves a group of clients with decentralized data who collaborate to learn a common model under the coordination of a centralized server, with the goal of protecting clients' privacy by ensuring that local datasets never leave the clients and that the server only performs model aggregation. However, in realistic scenarios, the server may be able to collect a small amount of data that approximately mimics the population distribution and has stronger computational ability to perform the learning process, resulting in the development of a hybrid FL framework. While previous hybrid FL work has shown that the alternative training of clients and server can increase convergence speed, it has focused on the scenario where clients fully participate and ignores the negative effect of partial participation. In this paper, we provide theoretical analysis of hybrid FL under clients' partial participation to validate that partial participation is the key constraint on the convergence speed. We then propose a new algorithm called FedCLG, which investigates the two-fold role of the server in hybrid FL. Firstly, the server needs to process the training steps using its small amount of local datasets. Secondly, the server's calculated gradient needs to guide the participating clients' training and the server's aggregation. We validate our theoretical findings through numerical experiments, which show that FedCLG outperforms state-of-the-art methods.
\end{abstract}

\begin{IEEEkeywords}
Federated Learning, Convergence Analysis, Server-Clients Collaboration
\end{IEEEkeywords}

\maketitle

\section{Introduction}


Recent years have seen exponential growth in data collection due to technological advancements, leading to the development of stronger machine learning models \cite{sarker2021machine}. However, traditional centralized machine learning algorithms struggle with handling the distributed nature of this type of data, which is often spread across multiple clients, such as mobile devices \cite{verbraeken2020survey}. To overcome this problem, Federated Learning (FL) \cite{mcmahan2017communication} has emerged as an important paradigm in modern machine learning. FL is a distributed machine learning approach where clients with decentralized data collaborate to learn a common model under the coordination of a parameter server. It has several advantages, including enhanced user data privacy \cite{xu2019hybridalpha, mothukuri2021survey}, scalability to new clients and datasets \cite{bonawitz2019towards}, and faster model convergence rate \cite{yu2019linear, yang2021achieving, bian2022mobility}. Despite these benefits, current FL systems typically assign the server to only simple computations, such as aggregating local models, wasting its powerful computational resources. Moreover, traditional FL assumes that datasets are exclusively available at the clients, either independently and identically distributed (IID) or non-IID. However, this is not always the case in real-world scenarios. In many cases, the entity building the machine learning model operates the server and possesses a small amount of data that approximately mimics the overall population distribution. Although a machine learning model can be trained based solely on the server data, the model performance will be limited by the size of the server dataset. Thus, a hybrid FL approach, which collaboratively utilizes the massive client data and a small amount of server data in a decentralized and privacy-preserving manner is of paramount practical importance to boost model performance.

Compared to traditional FL, which assumes that only clients can access data while the server can only perform model aggregation, the literature on hybrid FL is relatively scarce. Authors in \cite{augenstein2022mixed} make the assumption that the data collected by the server is complementary to the data held by each client. However, this assumption may only be applicable to specific scenarios, and in most real-world cases, the entity operating the server is likely to have access to a small amount of data that can approximate the population data distribution. \textcolor{black}{To address this issue, this paper adopts a similar setting to \cite{10001832}, which proposes a hybrid model training design called CLG-SGD (short for cascading
local-global SGD). In this design, the server performs aggregate-then-advance training. The empirical findings presented in \cite{10001832} demonstrate that compared to client-only local data training (e.g., Local-SGD), CLG-SGD enhances the convergence speed. However, its theoretical analysis bounds server-side and client-side updates separately, failing to comprehensively represent the theoretical benefits of additional server training in non-convex settings. Moreover, \cite{10001832} mainly focuses on the IID and fully-participated scenario, which may not be realistic in practical applications. In reality, clients may choose to participate only when they have access to a reliable Wi-Fi connection and a power source \cite{mcmahan2017communication}. Therefore, only a small percentage of clients may participate in each round. Additionally, data is distributed across multiple clients (e.g., mobile devices), each with its own unique data distribution \cite{zhao2018federated}. Thus, the investigation of hybrid FL under non-IID and partial participation scenarios is crucial.}


\begin{figure}[t]
	\centering	\includegraphics[width=0.99\linewidth]{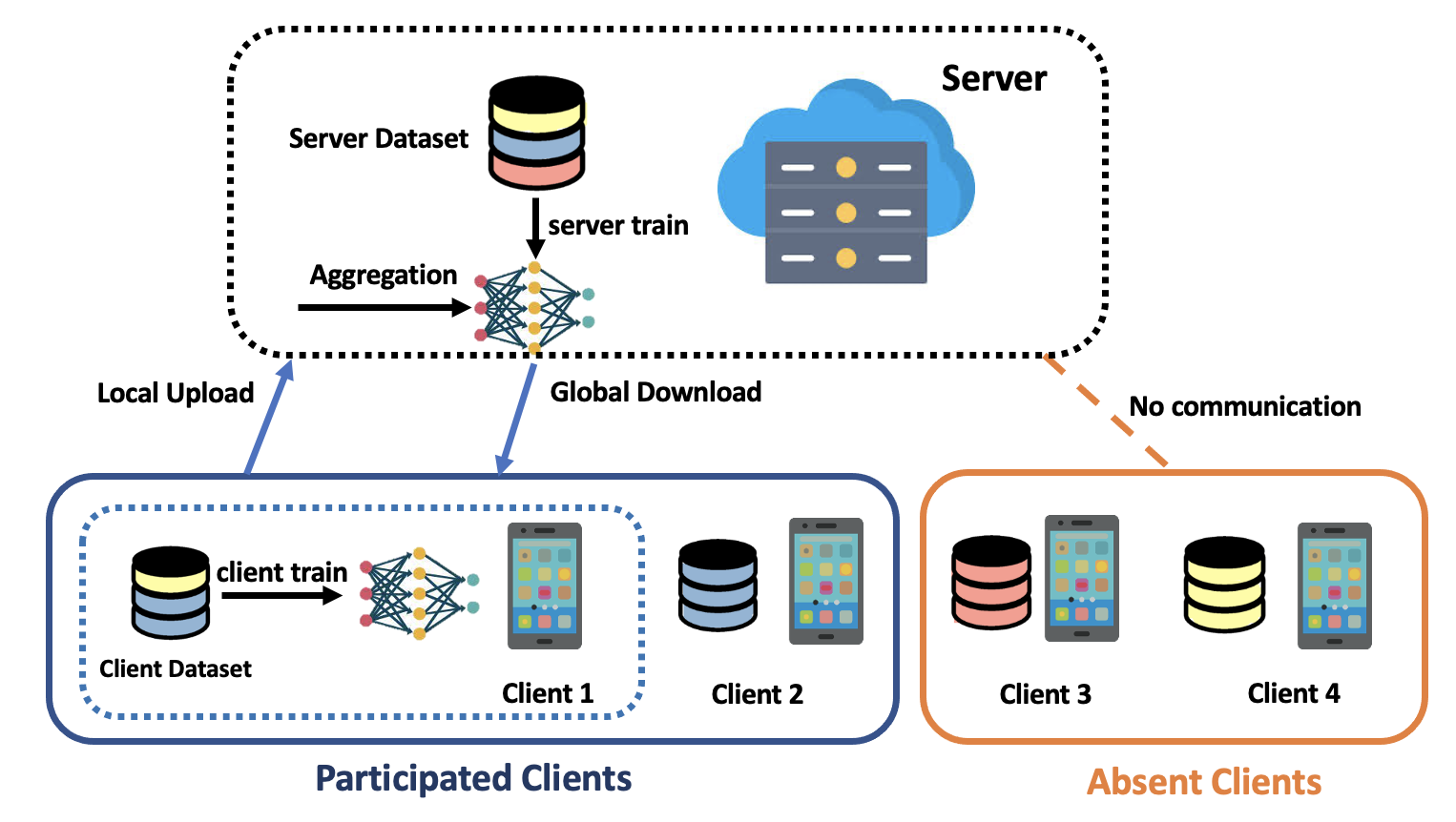}
	\caption{Illustration of a Communication Round in Hybrid FL with Non-IID Client Data and Partial Participation. } \label{fig:Hybrid_fl}
\vspace{-0.2in}
\end{figure}

\textcolor{black}{In this paper, we first revisit CLG-SGD \cite{10001832} under non-IID and partial participation setting, and introduce a novel convergence analysis. This analysis improves the convergence rate of CLG-SGD, demonstrating how additional server-side training can expedite convergence. However, our findings also indicate that partial participation errors can still impede CLG-SGD's convergence rate, even with augmented server training. To mitigate this issue, we then introduce FedCLG (Federated cascading
local-global learning), a new algorithm for hybrid FL that leverages server training to improve model convergence speed and correct partial participation errors in non-IID and partially participated scenarios (See Fig. \ref{fig:Hybrid_fl}). Specifically, FedCLG has two main responsibilities for server training. First, the server training starts with the latest aggregated global model and advances it using its limited local dataset. This allows the server to contribute to the global model with its advanced computation capabilities. Second, the server conducts an additional one-round training before broadcasting the new global model to each participating client. The gradient computed during this additional server training is utilized to correct partial participation errors. We propose two versions of FedCLG based on where partial participation errors are corrected: FedCLG-S and FedCLG-C. FedCLG-S corrects partial participation errors during server model aggregation, while FedCLG-C corrects them at each client's side during local training. Our proposed algorithm FedCLG aims to maximize the benefits of server training to improve model convergence speed and correct partial participation errors in non-IID and partially participated scenarios. We summarize our main contributions below:}

\begin{itemize}
    \item We provide a novel theoretical convergence analysis of the state-of-the-art hybrid FL method, CLG-SGD, that validates the benefit of additional server training without requiring the assumption of IID data or full client participation. Our analysis highlights that, despite the additional server training, convergence speed is still limited by partial participation errors.
    \item We propose FedCLG, a new algorithm that maximizes the potential benefits of server training in hybrid FL. We introduce two versions of FedCLG, FedCLG-S and FedCLG-C, to account for different communication and computation scenarios. We provide theoretical convergence analysis for both versions.
    \item We conduct extensive experiments on \textcolor{black}{three} datasets, demonstrating the superior performance of FedCLG over existing state-of-the-art methods.
\end{itemize}

The remainder of this paper is organized as follows. Related works are surveyed in Section II. The system model and extensive theoretical analysis of state-of-the-art hybrid FL method are presented in Section III. FedCLG is detailed in Section IV and its two version FedCLG-S and FedCLG-C are analyzed in Section V. Experiment results are reported in Section VII. Finally, Section VIII concludes the paper. All technical proofs can be found in the Appendices.

\section{Related Works}
\subsection{Federated Learning}
With the growing demand for local data storage and on-device model training, Federated Learning \cite{konevcny2016federated, liu2020deep, liu2020federated} has attracted significant interest in recent years. The Federated Averaging algorithm (FedAvg), first proposed by \cite{mcmahan2017communication}, operates by periodically averaging local Stochastic Gradient Descent (SGD) updates. This work has inspired numerous follow-up studies focusing on FL with IID client datasets and full client participation \cite{stich2018local, stich2019error, wang2021cooperative, cao2020fltrust}. Under the assumptions of complete participation and IID client datasets, several theoretical works \cite{Yu2018ParallelRS, stich2018local} have emerged, providing a linear speedup convergence guarantee, on par with the rate of parallel SGD \cite{zhang2016parallel}. However, real-world scenarios often present challenges in FL due to non-IID data and partial client participation \cite{zhu2021federated}. Recent works \cite{karimireddy2020scaffold, yang2021achieving, wang2022unified, li2022federated, li2023honest} have addressed these issues by offering similar convergence rates under non-IID and partial participation settings.

\subsection{Hybrid Federated Learning}
The majority of existing FL research focuses on the scenario where data is exclusively stored on the client side, and the server is only responsible for the aggregation step, ensuring clients' privacy requirements are met. However, this approach could potentially underutilize the server's computational capabilities. Compared to the clients, which are typically mobile devices in FL settings, the server generally possesses significantly greater computational power \cite{bahl2012advancing}. This has led to the emergence of a new FL configuration, referred to as hybrid FL. Current hybrid FL research can be divided into two categories, based on the source of the server dataset.

The first category of hybrid FL assumes that the server cannot collect data independently, while clients with limited computational resources can upload less privacy-sensitive data samples to the server to aid training \cite{elbir2022hybrid, huang2022wireless}. These studies concentrate on optimizing the trade-off between data sample communication costs and the benefits of model training.

The second line of hybrid FL, more closely related to this paper, assumes that the server can collect a small portion of the total data samples \cite{10001832, augenstein2022mixed}. While \cite{augenstein2022mixed} posits that the server's data complements each client's data, a more realistic assumption is that the server is more likely to gather a small amount of data that approximates the population data distribution \cite{jeong2020federated}. Our work adopts a similar setting to \cite{10001832}. However, while \cite{10001832} focuses on IID data and full client participation, our research investigates the more realistic scenario of non-IID client data and partial client participation.

\subsection{Variance Reduction}
Variance reduction has been a widely studied concept across various fields. Monte Carlo sampling methods employ the control variates technique to reduce variance \cite{glasserman2004monte}, while in stochastic gradient estimates for large-scale machine learning, SVRG \cite{johnson2013accelerating} and SAG \cite{reddi2016stochastic} have been introduced to reduce the stochastic sampling-variance. SAG was later simplified, leading to the proposal of SAGA \cite{defazio2014saga}. In Federated Learning (FL), the variance caused by randomly participating clients has a stronger impact than the variance caused by stochastically selected data samples. As a result, variance reduction methods in FL focus more on reducing client-variance. SCAFFOLD \cite{karimireddy2020scaffold}, an extension of SAGA, was proposed as the first variance reduction method in FL, which inspired subsequent works such as \cite{jhunjhunwala2022fedvarp, gu2021fast, acar2021federated, zhang2021fedpd} that attempt to reduce client-variance to increase convergence speed. However, none of these methods consider the hybrid FL setting and can result in the misutilization of stale information. In this work, we propose the first approach to reducing client-variance in hybrid FL, which fully exploits the benefits of server-side small datasets. A detailed comparison between our method and existing variance-reduction FL approaches is presented in Section IV.

\section{Problem Formulation and CLG-SGD}
\subsection{Problem Formulation}
In the hybrid federated learning setting, we aim to optimize the model parameters $x \in \mathcal{R}^d$ by minimizing the global objective function $f(x)$, similar to traditional federated learning frameworks. The global objective function is defined as:

\begin{align}
\label{c_opt}
\min_x f(x) = \frac{1}{N}\sum_{i=1}^N f_i(x),
\end{align}

where $f_i(x) = \frac{1}{m_i}\sum_{z \in \mathcal{D}_i} l(x,z)$ represents the local objective of client $i$ computed on their local dataset $\mathcal{D}_i$ with $m_i$ data points. The loss function is denoted by $l(.,.)$, and $z$ represents a data sample from the local dataset $\mathcal{D}_i$. The total data samples in the FL system are represented as $m$, such that $\sum_i^N m_i = m$, and the total number of clients is denoted as $N$. The underlying data distribution of the total $m$ data samples is denoted as $\mathcal{V}$. We assume, without loss of generality, that all $N$ clients' local objectives have equal weight in the global objective function (\ref{c_opt}). The algorithms and theoretical analysis can be easily extended to cases where client objectives are unequally weighted, such as proportional to the local data size.

In contrast to traditional FL, the hybrid FL framework posits that, in addition to the data available at each client, the server can collect a small dataset $\mathcal{D}_s^t$ with a constant size of $m_s$, \textcolor{black}{which is data homogeneous with the overall dataset.} Although the underlying data distribution of $\mathcal{D}_s^t$ remains consistent and approximates the overall population distribution $\mathcal{V}$, the dataset itself changes with each global round $t$ (while remaining fixed within the global round). Consequently, the server's optimization problem becomes:

\begin{align}
\label{s_opt}
\min_x f_s(x) = \frac{1}{m_s}\sum_{z \in \mathcal{D}_s^t} l(x,z).
\end{align}

However, because the size of $\mathcal{D}_s^t$ is considerably smaller than the overall dataset stored at each client (i.e., $m_s \ll m$), relying solely on $\mathcal{D}_s^t$ for model training could result in suboptimal outcomes. Additionally, the server's limited access to the fixed dataset for local training during specific time periods in each global round may significantly increase the training time.

Determining the best approach to utilize both server and clients' data for training and achieve optimal convergence performance is a challenging problem. To address this, we revisit the CLG-SGD algorithm in the hybrid FL setting. At each round $t$, the server randomly selects a subset of $M$ clients, denoted as $\mathcal{S}_t$, and sends the global model $x_t$ to these clients. Upon receiving $x_t$, each selected client $i$ performs $K$ rounds of local updates as follows:
\begin{align}
\label{c_local}
&x_{t,0}^i = x_t; \nonumber\\
&x_{t,k+1}^i = x_{t,k}^i - \eta g_{t,k}^i, \quad k = 0, \dots, K-1,
\end{align}
where $\eta$ is the client-side local learning rate, and $g_{t,k}^i = \nabla f_i(x_{t,k}^i, \zeta_i)$ is the stochastic gradient evaluated on a randomly drawn mini-batch $\zeta_i$ at client $i$ ($g_{t,k}^i = \nabla f_i(x_{t,k}^i)$ if a full gradient is used). After $K$ steps of local training, client $i$ sends back its update $\Delta_t^i = x_{t,K}^i - x_t$ to the server, which aggregates the updates to update the global model as follows:
\begin{align}
\label{s_aggregate}
x_{t+1}^s = x_t + \eta_g \frac{1}{M} \sum_{i\in \mathcal{S}_t} \Delta_t^i,
\end{align}
where $\eta_g$ is the global (aggregation) learning rate, and $x_{t+1}^s$ represents an intermediate stage between client local training and server local training. In classic FL, the iteration ends at this point. However, in hybrid FL, the server not only aggregates the clients' updates but also utilizes its own dataset $\mathcal{D}_s^t$ for server training. Thus, after aggregating the model $x_{t+1}^s$, the server also performs $E$ rounds of local updates as follows:
\begin{align}
\label{s_local}
&x_{t+1,0}^s = x_{t+1}^s; \nonumber\\
&x_{t+1,e+1}^s = x_{t+1,e}^s - \gamma g_{t+1,e}^s, \quad e = 0, \dots, E-1; \nonumber\\
&x_{t+1} = x_{t+1,E}^s,
\end{align}
where $\gamma$ is the server learning rate, and $g_{t,e}^s = \nabla f_s(x_{t,e}^s, \zeta_s)$ is the stochastic gradient evaluated on a randomly drawn mini-batch $\zeta_s$ from the server dataset $\mathcal{D}_s^t$ ($g_{t,e}^s = \nabla f_s(x_{t,e}^s)$ if a full gradient is used). After the additional model training at the server, the global model advances from $x_{t+1}^s$ to $x_{t+1}$. Then the server broadcasts the new global model $x_{t+1}$ for the next round of iteration. 

\textcolor{black}{In the previous work \cite{10001832}, the authors focus on IID and fully participating settings and fail to show how additional server training accelerates the convergence speed in non-convex settings. In this paper, we extend the analysis to non-IID and partial participation settings in the following subsection. Our novel theoretical analysis demonstrates that, although additional server training can improve the convergence rate, convergence speed is still dominated by partial participation error resulting from data heterogeneity and randomly selected clients. Based on this observation, we propose a new algorithm, FedCLG, in Section IV.}

\subsection{Novel Convergence Analysis of CLG-SGD \cite{10001832}}
\textcolor{black}{
For the theoretical analysis in this paper, we make the following assumptions: in each round, the server selects a subset of clients uniformly without replacement. In addition, our convergence analysis will utilize the following standard technical assumptions.
}
\begin{assumption}[Lipschitz Smoothness]\label{assm:smooth}
There exists a constant $L>0$ such that $\|\nabla f_i(x) - \nabla f_i(y)\| \leq L\|x- y\|$, $\forall x, y \in \mathbb{R}^d$ and $\forall i = 1, ..., N$. 
\end{assumption}

\begin{assumption}[Bounded Variance]\label{assm:unbiased-local}
The dataset $\mathcal{D}_s^t$ at the server approximates the overall population distribution $\mathcal{V}$, so the gradient calculated using $\mathcal{D}_s^t$ is an unbiased estimate of the global objective, i.e., $\mathbb{E}_{\mathcal{D}_{s}^t \sim \mathcal{V}}[\nabla f_{s}(x)] = \nabla f(x)$. Furthermore, there exists a constant $\sigma > 0$ such that the variance of the gradient estimator is bounded, i.e.,
\begin{align}
\mathbb{E}_{\mathcal{D}_{s}^t \sim \mathcal{V}}\left[\|\nabla f_{s}(x) - \nabla f(x)\|^2 \right] \leq \frac{\sigma^2}{m_s}, \forall x, \forall t.
\end{align}
where $m_s$ is the size of server dataset $\mathcal{D}_s^t$.
\end{assumption}

\textcolor{black}{
\begin{assumption}[Unbiased Gradient Estimate and Bounded Local Variance] \label{l_variance}
The stochastic gradient estimate is unbiased, i.e., $\mathbb{E}_\zeta[ F_i(x, \zeta)] = \nabla f_i(x)$, $\forall x$ and $\forall i = 1, \cdots, N$ and its variance is bounded $\mathbb{E}[\|\nabla F_i(x, \zeta_i) - \nabla f_i(x)\|^2] \leq \sigma_l^2$, $\forall x \in \mathbb{R}^d$ and $\forall i = 1, ..., N$.
\end{assumption}}

\begin{assumption} [Bounded Global Variance]\label{a_variance}
	There exists a constant number $\sigma_g > 0$ such that the variance between the local gradient of client $i$ and the global gradient is bounded as follows:
 \begin{align}
     \| \nabla f_i(x) - \nabla f(x) \|^2 \leq \sigma_{g}^2, ~\forall i \in [N], \forall x.
 \end{align}
\end{assumption}

Assumptions \ref{assm:smooth}, \ref{l_variance} and \ref{a_variance} are commonly adopted in the convergence analysis of FL under non-IID settings \cite{yang2021achieving, bian2022mobility, yu2019linear, Yu2018ParallelRS}. Assumption \ref{assm:unbiased-local} provides a bound on the variance introduced by server local training, which is dependent on the size of $\mathcal{D}_s^t$ \cite{prakash2021talk}. We here consider the size of $\mathcal{D}_s^t$ as a hyper-parameter.

\textcolor{black}{
\begin{theorem} \label{th1}
Suppose that client local learning rate $\eta$, global learning rate $\eta_g$, and server local learning rate $\gamma$ are chosen such that $\eta \leq \frac{1}{3KL}$, $\eta \eta_g \leq \frac{1}{27KL}$, and $\gamma \leq \frac{1}{6EL}$. Under Assumptions \ref{assm:smooth}, \ref{assm:unbiased-local}, \ref{l_variance}, \ref{a_variance}, suppose that in each round $t$ the server uniformly selects $M$ out of $N$ clients without replacement, the sequence of model vectors ${x_t}$ satisfies:
\begin{align}
    & \min_{t \in [T]} \mathbb{E}\|\nabla f(x_t)\|_2^2 = \mathcal{O}\bigg(\frac{(f_0 - f_*)}{T(\gamma E + \eta\eta_g K)}\bigg) \nonumber\\
     & + \mathcal{O}\bigg(\frac{\eta^3\eta_g L^2 K^3\sigma_g^2}{\gamma E + \eta\eta_g K}\bigg)  +  \mathcal{O}\bigg(\frac{\gamma^2 E L \sigma^2}{m_s(\gamma E + \eta\eta_g K)} \bigg) \nonumber\\
     & +  \mathcal{O}\bigg(\frac{(N-M)K^2\eta^2\eta_g^2L\sigma_g^2}{M(N-1)(\gamma E + \eta\eta_g K)}\bigg)  + \mathcal{O}\bigg(\frac{\eta^3\eta_g L^2 K^2\sigma_l^2}{\gamma E + \eta\eta_g K}\bigg) \nonumber\\
     & +\mathcal{O}\bigg(\frac{\eta^2\eta_g^2LK \sigma_l^2}{M(\gamma E + \eta\eta_g K)}\bigg),
\end{align}
\end{theorem}}

\begin{proof}
The proof is shown in Appendix A.
\end{proof}
\textcolor{black}{
\begin{Remark}
The convergence bound presented above consists of six terms, with the second term accounting for the effect of client local training, the third term representing the impact of limited data availability at the server, the fourth term reflecting the influence of partial participation of clients, and the last two terms representing the error caused stochastic client updates.
\end{Remark}
}

\begin{Remark}
The third term is influenced by the number of data points available at the server, denoted by $m_s$. As $m_s$ increases, the convergence bound tightens, which aligns with the expectation that having more training data stored at the server should result in better convergence performance.
\end{Remark}
\textcolor{black}{
\begin{Corollary}
Let $\eta = \Theta(\frac{1}{K\sqrt{T}})$, $\eta_g = \Theta(\sqrt{MK})$ and $\gamma = \Theta(\frac{1}{\sqrt{ET}})$, the convergence rate of CLG-SGD becomes:
\begin{align}
     & \min_{t \in [T]} \mathbb{E}\|\nabla f(x_t)\|_2^2 = \mathcal{O}\bigg(\frac{\sqrt{MK}}{(\sqrt{MK}+\sqrt{E})T}\bigg) \nonumber\\
     & + \mathcal{O}\bigg(\frac{K}{(\sqrt{MK}+\sqrt{E})\sqrt{T}}\bigg) 
\end{align}
\end{Corollary}
}

\begin{Remark}
Our convergence analysis of CLG-SGD is more general than that of \cite{10001832}, as we consider non-IID data and partial client participation. Additionally, our analysis demonstrates the theoretical benefits of additional server training even in the case of IID data and full participation, where \cite{10001832} fails to do so. Notably, in the case of IID data ($\sigma_g = 0$) or full participation ($M=N$), our convergence rate is $\mathcal{O}\bigg(\frac{1}{(\sqrt{MK}+\sqrt{E})\sqrt{T}}\bigg) + \mathcal{O}\bigg(\frac{\sqrt{MK}}{(\sqrt{MK}+\sqrt{E})T}\bigg)$, which converges faster than the rate of $\mathcal{O}\bigg(\frac{1}{\sqrt{MKT}}\bigg) + \mathcal{O}\bigg(\frac{1}{T}\bigg) $ found in \cite{10001832}.
\end{Remark}

\textcolor{black}{
\begin{Remark}
If we consider full client participation and set the server's local training epoch $E=0$, the hybrid FL approach becomes equivalent to classic FL, and the convergence speed degenerates to $\mathcal{O}\bigg(\frac{1}{\sqrt{MKT}}\bigg) + \mathcal{O}\bigg(\frac{1}{T}\bigg)$. This rate is the same as the state-of-the-art convergence rate found in classical FL \cite{karimireddy2020scaffold, yang2021achieving} without server training.
\end{Remark}
}

\begin{Remark}
The corollary reveals that the dominating factor in the convergence bound is $\mathcal{O}\bigg(\frac{K}{(\sqrt{MK}+\sqrt{E})\sqrt{T}}\bigg)$, which is closely related to the global variance $\sigma_g^2$. This suggests that the global variance has a more significant effect on convergence behavior in cases with partial participation, particularly in highly non-IID scenarios where $\sigma_g$ is substantial. Therefore, developing a new hybrid FL approach to mitigate the negative effects of partial participation in non-IID settings is a challenging task.
\end{Remark}

Our novel theoretical analysis of CLG-SGD suggests that hybrid FL can achieve faster convergence by incorporating additional server local training after the aggregation step. However, like classic FL, hybrid FL is still constrained by the convergence limitations caused by partial client participation in non-IID settings. Prior research, such as \cite{karimireddy2020scaffold, jhunjhunwala2022fedvarp}, has used variance reduction techniques in classic FL to mitigate the adverse effects of partial participation. Although these methods can be adapted to hybrid FL, they do not fully leverage the potential benefits of the small amount of server data. Therefore, it is critical to devise an algorithm that fully exploits the server data. In the next sections, we introduce FedCLG, a novel algorithm that fully exploits the server data, and we demonstrate its superior convergence performance.

\section{FedCLG}
In hybrid FL, a significant difference from the classical FL setting is the server's possession of its own local training dataset $\mathcal{D}_s^t$, which is a small subset approximating the overall population dataset. While using only the server dataset $\mathcal{D}_s^t$ to train a model has drawbacks, such as slower training speed and higher risk of reaching sub-optimal points, it can provide a more accurate direction of the global objective compared to using the large non-IID dataset stored at each client. The key innovation of FedCLG is the utilization of the server gradient to produce variance correction either at the server aggregation step (FedCLG-S) or client local training step (FedCLG-C). This correction helps address the issue of non-IID data distribution across clients and ultimately improves the FL model's accuracy. In this section, we provide further elaboration on the specific details of each version of FedCLG.

\subsection{FedCLG-C}
In FedCLG-C, the server randomly selects a subset of $M$ clients out of $N$ total clients, denoted as $\mathcal{S}_t$, at each round $t$. Prior to broadcasting the global model $x_t$ to the selected clients, the server conducts an additional local training step using its own local dataset $\mathcal{D}_s^t$ based on the global model $x_t$, producing a gradient denoted as $g_t^s$, where $g_t^s = \nabla f_s(x_t)$ represents the full batch gradient or $g_t^s = \nabla f_s(x_t,  \zeta_s)$ represents the stochastic gradient evaluated on a randomly drawn mini-batch $\zeta_s$ from the server dataset $\mathcal{D}_s^t$. FedCLG-C requires the server to broadcast both the global model $x_t$ and the gradient $g_t^s$ to each selected client $i \in \mathcal{S}_t$. Upon receiving the gradient $g_t^s$ and the model $x_t$, each client $i$ performs $K$ rounds of local epoch with the correction term $c_i$ as follows:
\begin{align}
    & c_i = g_t^s - g_t^i \nonumber\\
    &x_{t,0}^i = x_t; \nonumber\\
    &x_{t,k+1}^i = x_{t,k}^i - \eta (g_{t,k}^i + c_i), ~k = 0, \dots, K-1,
\label{FedCLG-C}
\end{align}
\textcolor{black}{Here, $g_{t}^i = \nabla f_i(x_{t}, \zeta_i)$ is the stochastic gradient evaluated on a randomly drawn mini-batch $\zeta_i$ at client $i$. ($g_{t}^i = \nabla f_i(x_{t})$ if a full gradient is used).} After completing $K$ rounds of training, each client $i$ sends its update to the server. The server then carries out the same aggregation and server's local training steps as described in Eqs. \ref{s_aggregate} and \ref{s_local}. We summarize FedCLG-C in Algorithm 1.


\subsection{FedCLG-S}
In FedCLG-S, similar to FedCLG-C, the server selects a random subset of $M$ clients at each round $t$ and performs an extra training step based on $x_t$ using its own local dataset $\mathcal{D}_s^t$. The resulting training gradient $g_t^s$ is held by the server, which subsequently broadcasts the global model $x_t$ to the selected clients. Upon receipt of the global model, each selected client $i \in \mathcal{S}_t$ conducts $K$ steps of local training as specified by Eq. \ref{c_local} and calculates the client gradient $g_t^i$ based on the received global model. Each client sends back its cumulative local updates and local gradient $g_t^i$ to the server, which then aggregates the updates using the following formula:
\begin{align}
x_{t+1}^s = x_t + \eta_g \frac{1}{M} \sum_{i\in \mathcal{S}_t} (\Delta_t^i - K\eta(g_t^s - g_t^i)),
\label{FedCLG-S}
\end{align}
where $\eta_g$ is the learning rate for the server's local training. After this aggregation, the server performs $E$ epochs of local training as described in Eq. \ref{s_local}. The steps involved in FedCLG-S are summarized in Algorithm 1.



\begin{algorithm}[h]
    \caption{FedCLG}
    \label{alg:FedCLG}
    \begin{algorithmic}[1]
        \STATE{Initial model $x_0$, client local learning rate $\eta$, global learning rate $\eta_g$, server local learning rate $\gamma$, number of client local epoch $K$, number of server local epoch $E$, number of global iterations $T$}
        \FOR{$t = 0, 1, \dots, T-1$}
        \STATE{Uniformly sample $\mathcal{S}_t$ clients without replacement}
        \STATE{Compute a server gradient $g_t^s = \nabla f_s(x_t,  \zeta_s)$}
        \STATE{\textit{Client Side:}}
        \FOR{each client $i \in \mathcal{S}_t$ in parallel}
            \IF{FedCLG-C}
                \STATE{Perform client local training as Eq. \ref{FedCLG-C}}
            \ELSIF{FedCLG-S}
                \STATE{Perform client local training as Eq. \ref{c_local}}
            \ENDIF
        \ENDFOR
        \STATE{\textit{Server Side:}}
        \IF{FedCLG-C}
                \STATE{Aggregate the model $x_{t+1}^s$ as Eq. \ref{s_aggregate}}
        \ELSIF{FedCLG-S}
                \STATE{Aggregate the model $x_{t+1}^s$ as Eq. \ref{FedCLG-S}}
            \ENDIF
        \STATE{Perform server local training as Eq. \ref{s_local}}
        \ENDFOR
    \end{algorithmic}
\end{algorithm}

The primary difference between FedCLG-C and FedCLG-S, which is shown in Fig. \ref{fig:FedCLG}, lies in their methods for addressing the variance issue. FedCLG-C addresses the problem of partial participation during local training on the client side, while FedCLG-S corrects the partial participation error at the server side during the aggregation step. While FedCLG-C requires that the server broadcast an additional gradient $g_t^s$ during communication, FedCLG-S requires that each client upload an additional gradient $g_t^i$ to the server per round. The choice between using FedCLG-S or FedCLG-C should be based on the available bandwidth for uploading and downloading. Specifically, if the download bandwidth is restricted, FedCLG-S should be utilized. Conversely, if the upload bandwidth is restricted, FedCLG-C should be used. \textcolor{black}{To further address the communication efficiency concerns, our method can be integrated with quantization or compression techniques, which have been extensively studied in the FL setting (e.g.\cite{reisizadeh2020fedpaq, shlezinger2020uveqfed, jhunjhunwala2021adaptive, mao2022communication}). These methods are capable of significantly reducing the communication costs associated with additional transmissions.}

\begin{figure}[t] 
	\centering
	\includegraphics[width=0.99\linewidth]{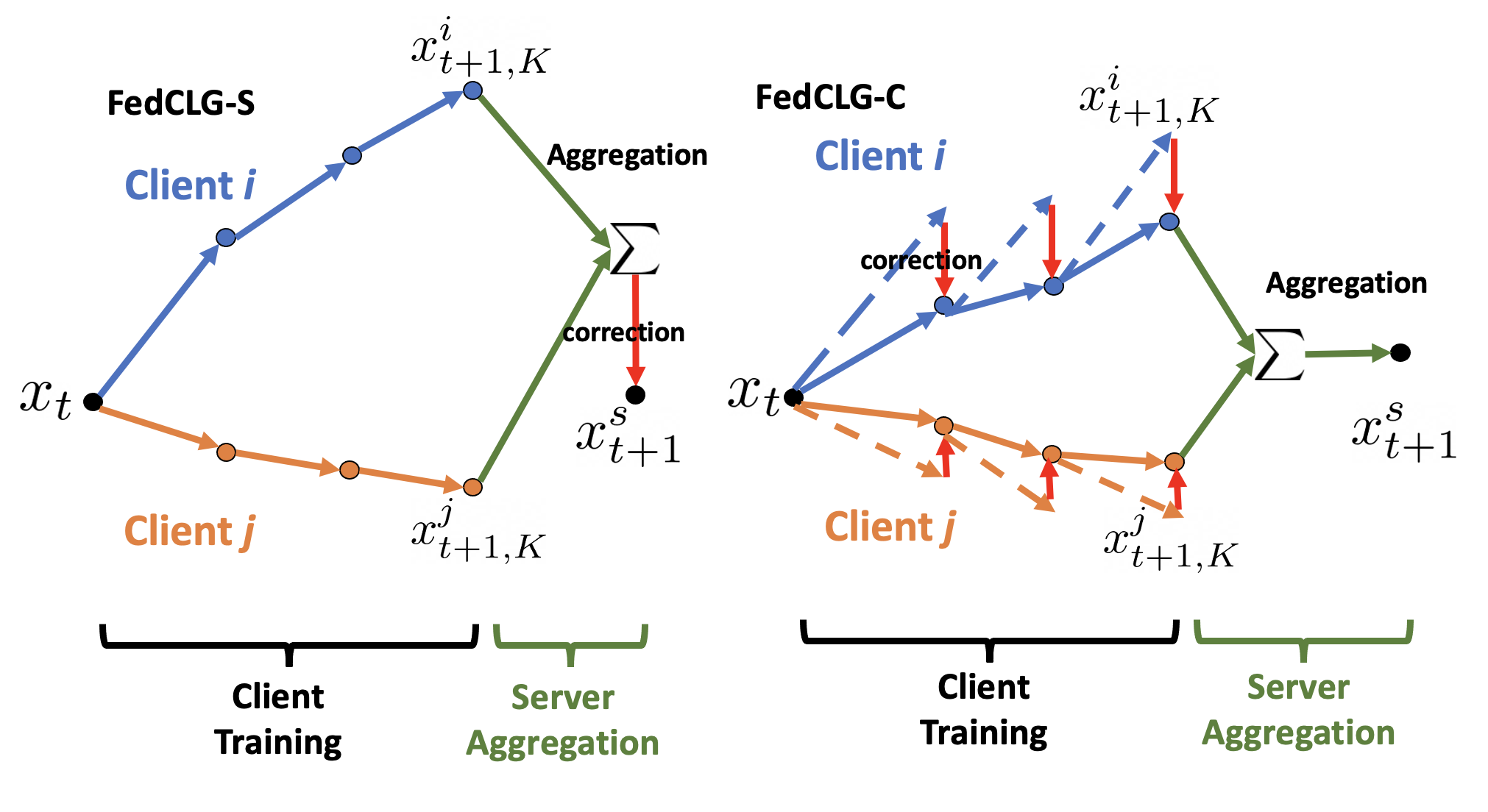}
	\caption{The key distinction between FedCLG-S and FedCLG-C lies in the timing and location of the correction step. In FedCLG-S, the corrections occur during the server aggregation step, whereas in FedCLG-C, they take place during each client's local training step.} \label{fig:FedCLG}
\vspace{-0.2in}
\end{figure}

\subsection{Comparison with FL Variance Reduction Methods}
Existing variance reduction methods in FL do not consider the potential benefits of using the server dataset to reduce variance. SCAFFOLD, proposed in \cite{karimireddy2020scaffold}, is the first work to identify client drift error and utilize control variates to correct it. However, SCAFFOLD requires additional gradient communication during both the upload and download processes. Alternatively, FedCLG-C or FedCLG-S can be chosen based on different upload/download communication scenarios, reducing the overall communication workload. Other variance reduction methods, such as \cite{jhunjhunwala2022fedvarp} and \cite{gu2021fast}, require the server to maintain $\mathcal{O}(Nd)$ memory, where $N$ is the number of total clients and $d$ is the model size, which can be very expensive and unrealistic in cross-device settings of FL. \textcolor{black}{Others \cite{zhang2021fedpd, acar2021federated} require additional client computations.} Moreover, all of the above methods use stale information to build the correction term $c_i$, which can negatively affect performance. Furthermore, none of these methods provide convergence guarantees under the hybrid FL setting. In the experimental section, we demonstrate the superiority of FedCLG.

\section{Convergence Analysis of FedCLG}
In this section, we will provide a convergence analysis of both versions of FedCLG in a non-convex setting. We will adopt the same assumptions as in the previous section, which were used for the convergence analysis of CLG-SGD.
\textcolor{black}{
\begin{theorem} \label{th2}
Suppose that client local learning rate $\eta$, global learning rate $\eta_g$ and server local learning rate $\gamma$ are chosen such that $\eta \leq \frac{1}{8KL}$, $\eta \eta_g \leq \frac{1}{36KL}$ and $\gamma \leq \frac{1}{6EL}$. Under Assumptions \ref{assm:smooth}, \ref{assm:unbiased-local}, \ref{l_variance}, suppose in each round $t$ the server uniformly selects $M$ out of $N$ clients without replacement, the sequence of FedCLG-C model vectors ${x_t}$ satisfies:
\begin{align}
    & \min_{t \in [T]} \mathbb{E}\|\nabla f(x_t)\|_2^2 = \mathcal{O}\bigg(\frac{(f_0 - f_*)}{T(\gamma E + \eta\eta_g K)}\bigg)\ \nonumber\\
    & + \mathcal{O}\bigg(\frac{\eta^3\eta_g L^2 K^3\sigma^2}{m_s(\gamma E + \eta\eta_g K)}\bigg) + \mathcal{O}\bigg(\frac{\gamma^2L E \sigma^2}{m_s(\gamma E + \eta\eta_g K)} \bigg) \nonumber\\
    & + \mathcal{O}\bigg(\frac{\eta^2\eta_g^2KL\sigma^2}{M m_s(\gamma E + \eta\eta_g K)} \bigg)  + \mathcal{O}\bigg(\frac{\eta^3\eta_g L^2 K^2\sigma_l^2}{\gamma E + \eta\eta_g K}\bigg) \nonumber\\
    & +\mathcal{O}\bigg(\frac{\eta^2\eta_g^2LK \sigma_l^2}{M(\gamma E + \eta\eta_g K)}\bigg) ,
\end{align}
where $f_0 = f(x_0)$, $f_* = f(x_*)$.
\end{theorem}
}

\begin{proof}
The proof is shown in Appendix B.
\end{proof}


\begin{Corollary}
Let $\eta = \Theta(\frac{1}{K\sqrt{T}})$, $\eta_g = \Theta(\sqrt{MK})$ and $\gamma = \Theta(\frac{1}{\sqrt{ET}})$, the convergence rate of FedCLG-C becomes:
\begin{align}
     & \min_{t \in [T]} \mathbb{E}\|\nabla f(x_t)\|_2^2 \textcolor{black}{=} \mathcal{O}\bigg(\frac{1}{(\sqrt{MK}+\sqrt{E})\sqrt{T}}\bigg) \nonumber\\
     & + \mathcal{O}\bigg(\frac{\sqrt{MK}}{(\sqrt{MK}+\sqrt{E})T}\bigg)
\end{align}
\end{Corollary}
\textcolor{black}{
\begin{theorem} \label{th3}
Suppose that client local learning rate $\eta$, global learning rate $\eta_g$ and server local learning rate $\gamma$ are chosen such that $\eta \leq \frac{1}{3KL}$, $\eta \eta_g \leq \frac{1}{27KL}$ and $\gamma \leq \frac{1}{6EL}$. Under Assumptions \ref{assm:smooth}, \ref{assm:unbiased-local}, \ref{l_variance}, \ref{a_variance}, suppose that in each round $t$ the sever uniformly selects $M$ out of $N$ clients without replacement, the sequence of FedCLG-S model vectors ${x_t}$ satisfies:
\begin{align}
    & \min_{t \in [T]} \mathbb{E}\|\nabla f(x_t)\|_2^2 = \mathcal{O}\bigg(\frac{(f_0 - f_*)}{T(\gamma E + \eta\eta_g K)}\bigg)\ \nonumber\\
    & + \mathcal{O}\bigg(\frac{\eta^3\eta_g L^2 K^3\sigma_g^2}{\gamma E + \eta\eta_g K}\bigg) + \mathcal{O}\bigg(\frac{\gamma^2L E \sigma^2}{m_s(\gamma E + \eta\eta_g K)} \bigg) \nonumber\\
    & + \mathcal{O}\bigg(\frac{\eta^2\eta_g^2KL\sigma^2}{M m_s(\gamma E + \eta\eta_g K)} \bigg) + \mathcal{O}\bigg(\frac{\eta^3\eta_g L^2 K^2\sigma_l^2}{\gamma E + \eta\eta_g K}\bigg) \nonumber\\
    & +\mathcal{O}\bigg(\frac{\eta^2\eta_g^2LK \sigma_l^2}{M(\gamma E + \eta\eta_g K)}\bigg) ,
\end{align}
where $f_0 = f(x_0)$ and $f_* = f(x_*)$. 
\end{theorem}
}
\begin{proof}
The proof is shown in Appendix C.
\end{proof}

\begin{Corollary}
Let $\eta = \Theta(\frac{1}{K\sqrt{T}})$, $\eta_g = \Theta(\sqrt{MK})$ and $\gamma = \Theta(\frac{1}{\sqrt{ET}})$, the convergence rate of FedCLG-S becomes:
\begin{align}
     & \min_{t \in [T]} \mathbb{E}\|\nabla f(x_t)\|_2^2 \textcolor{black}{=} \mathcal{O}\bigg(\frac{1}{(\sqrt{MK}+\sqrt{E})\sqrt{T}}\bigg) \nonumber\\
     & + \mathcal{O}\bigg(\frac{\sqrt{MK}}{(\sqrt{MK}+\sqrt{E})T}\bigg)
\end{align}
\end{Corollary}

\textcolor{black}{
\begin{Remark}
 Both FedCLG-C and FedCLG-S's convergence results contain six terms where the second term is client local drift error, the third term is server training update error, the fourth term captures the stochastic error due to the limited size of the server's subset data compared to the overall population and the last two terms are stochastic client update error. Notably, both FedCLG-S and FedCLG-C eliminate the partial participation error. Additionally, the bound becomes increasingly dependent on the size of the dataset stored at the server, which is expected since we use the server gradient to guide client updates. A larger server dataset can result in a more accurate server gradient direction, leading to a tighter overall bound. While a smaller $m_s$ may produce less accurate corrections than a larger $m_s$, our experiments demonstrate that even with a small $m_s$, our proposed method can significantly enhance convergence speed under extremely non-IID settings.
\end{Remark}
}
\vspace{-4mm}
\textcolor{black}{
\begin{Remark}
The primary distinction between FedCLG-C and FedCLG-S lies in the second term, which addresses client local drift error. This is justifiable as FedCLG-C incorporates a correction term at each client's local training step, making it less dependent on global objective variance and more reliant on the quality of the correction step. Conversely, FedCLG-S applies the correction step during server aggregation, allowing client local training to be influenced by the variance between local and global objectives. This indicates that FedCLG-C may yield marginally better results when server datasets are reliable (i.e., low $\frac{\sigma^2}{m_s}$ ) and client data heterogeneity is high. However, FedCLG-S remains a robust choice, outperforming baseline methods (as shown in the experiment section). Under less extreme non-IID conditions, the choice between FedCLG-S and FedCLG-C should be influenced by bandwidth considerations.
\end{Remark}
}
\vspace{-4mm}
\textcolor{black}{
\begin{Remark}
    Under partial client participation setting, assume the server's local training epoch $E=0$, which reduces hybrid FL setting to classic FL setting, the convergence rates of FedCLG-C and FedCLG-S reduce to $\mathcal{O}\bigg(\frac{1}{\sqrt{MKT}}\bigg) + \mathcal{O}\bigg(\frac{1}{T}\bigg)$ which match the convergence rate achieved by the SOTA variance reduction methods \cite{karimireddy2020scaffold, jhunjhunwala2022fedvarp} used in the classic FL setting with partial client participation.
\end{Remark}
}

\vspace{-4mm}
\textcolor{black}{
\begin{Remark}
     To prevent client local drift error from dominating the convergence process, aiming for a convergence rate of $\mathcal{O}\left(\frac{1}{(\sqrt{MK}+\sqrt{E})\sqrt{T}}\right)$, both FedCLG-S and FedCLG-C need that the local epoch $K$ should not surpass $T/M$. 
\end{Remark}
}

\begin{Remark}
Both FedCLG-C and FedCLG-S exhibit convergence rates of $\mathcal{O}\bigg(\frac{1}{(\sqrt{MK}+\sqrt{E})\sqrt{T}}\bigg)$ under non-IID and partial participation settings, provided that there are enough training rounds $T$ \textcolor{black}{(i.e. $T \geq KM$)}. This rate is faster than that of CLG-SGD, which converges with a rate dominated by $\mathcal{O}\bigg(\frac{K}{(\sqrt{MK}+\sqrt{E})\sqrt{T}}\bigg)$. Interestingly, larger client-side local training epochs $K$ can actually hurt the convergence rate for CLG-SGD due to the negative effects of partial participation. However, after eliminating these negative effects in both FedCLG-S and FedCLG-C, the new convergence rates show that larger client-side local training epochs $K$ can actually increase the convergence rate.
\end{Remark}

\section{Experiments}
\subsection{Setup}
We conducted all experiments using Federated Learning (FL) simulation on the PyTorch framework and trained the models on Geforce RTX 3080 GPUs. We performed five random repeats and reported the averaged results. The detailed experimental settings are presented below.

\subsubsection{Dataset and Backbone Model}
To verify our theoretical findings, we evaluate the proposed methods on three datasets:

\textbf{MNIST \cite{deng2012mnist}}: We utilize LeNet-5 \cite{726791} as the backbone model. The default training hyperparameters are as follows: server local learning rate \textcolor{black}{tuning from $\gamma = \{0.01, 0.05, 0.25\}$}, client local learning rate \textcolor{black}{tuning from $\eta = \{0.01, 0.05, 0.25\}$}, local learning rates' decay factor equals to 0.99 until learning rate reaches $0.001$, global learning rate $\eta_g = 1$, and training batch sizes at both the server and clients set to 64.

\textbf{CIFAR-10 \cite{krizhevsky2009learning}}: We also utilize LeNet-5 as the backbone model. The default training hyperparameters are server local learning rate \textcolor{black}{tuning from $\gamma = \{0.01, 0.05, 0.25\}$}, client local learning rate \textcolor{black}{tuning from $\eta = \{0.02, 0.08, 0.32\}$}, local learning rates' decay factor equals to 0.99 until learning rate reaches $0.001$, global learning rate $\eta_g = 1$, and training batch sizes at both the server and clients set to 128.

\textcolor{black}{
\textbf{CIFAR-100 \cite{krizhevsky2009learning}}: For CIFAR-100 datasets, we use 20 superclasses to reclassify the data samples. We utilize MobileNetV2 \cite{sandler2018mobilenetv2} as the backbone model. The default training hyperparameters are server local learning rate \textcolor{black}{tuning from $\gamma = \{0.005, 0.05, 0.5\}$}, client local learning rate \textcolor{black}{tuning from $\eta = \{0.01, 0.1, 1\}$}, local learning rates' decay factor equals to 0.99 until learning rate reaches $0.001$, global learning rate $\eta_g = 1$, and training batch sizes at both the server and clients set to 128.
}

\subsubsection{Client Setting}
Our experimental evaluations cover both IID and non-IID client datasets. Without loss of generality, we assume that each client has an equal number of data samples. Specifically, for MNIST, we simulate 200 clients with 150 data samples each, for CIFAR-10 and CIFAR-100 we simulate 200 clients with 200 data samples each. 

\textbf{IID and Non-IID scenario.}
For IID datasets, we randomly assign an equal-size local dataset from the total training set to each client. For non-IID datasets, we apply the Dirichlet method ($\alpha$) to create the data distribution for each client. We use various $\alpha$ values to demonstrate the convergence performance under different degrees of non-IID.

\textbf{Client participation.}
The primary objective of this paper is to investigate how the server dataset can be leveraged to eliminate the partial participation error in Federated Learning. Therefore, in the experiment section, we will focus on the scenario where clients participate partially in the training process. Specifically, in the main experiments, we uniformly sample $M = 4$ clients without replacement in every round for MNIST dataset, and $M=10$ clients for CIFAR-10 and CIFAR-100 dataset.

\textbf{Number of Local Epoch $K$.}
To investigate the impact of the client's local training epoch, we vary the value of $K$ to be $1$, $3$, or $5$.

\subsubsection{Server Setting}
Compared to classical FL, Hybrid FL introduces the novel setting where the server itself contains a subset of the population dataset. 

\textbf{Size of server dataset $m_s$.}
In the MNIST experiment, we consider the server dataset to contain $1\%$ of the total training data samples, while in the CIFAR-10 and CIFAR-100 experiments, we consider it to contain $5\%$ of the total training data samples. In the ablation studies, we change the size of the server dataset to examine its impact on the hybrid FL approach.

\textbf{Number of Local Epoch $E$.}
To investigate the impact of the server's local training epoch, we vary the value of $E$ to be $1$, $3$, or $5$.

\subsubsection{Baselines}
In our study, we evaluate the performance of our approach against the following established methods: (1) \textbf{Server-only}: This method involves training a model solely on the server's dataset, denoted as $\mathcal{D}_s^t$, without using any client data.
(2) \textbf{FedAvg}: As a key baseline in traditional Federated Learning (FL), FedAvg employs client-side data for distributed learning. It is worth noting that the server's dataset is not used in this method.
(3) \textbf{CLG-SGD}: This state-of-the-art hybrid FL technique, introduced in \cite{10001832}, alternates between server local training and client local training.
(4) \textbf{SCAFFOLD+}: SCAFFOLD \cite{karimireddy2020scaffold} is a leading FL variance reduction method. To enable a fair comparison that highlights the advantages of our proposed approach, we have incorporated the alternating training concept from CLG-SGD into SCAFFOLD. Consequently, SCAFFOLD+ is an improved version of SCAFFOLD that utilizes the server's local dataset for additional training.

\subsection{Experiments Results}
The primary objective of the experiments is to showcase the differences in the number of global training rounds required by different methods to achieve a specific test accuracy, thereby highlighting the differences in their convergence speeds.

\textbf{Performances Comparison.} 
We first compare the convergence performance of our proposed methods and baselines on the MNIST dataset. We consider both IID and non-IID settings. As shown in Fig. \ref{MNIST_main} (a), for the IID setting, FedCLG-S and FedCLG-C outperform FedAvg and Server-Only, but only achieve comparable performance with CLG-SGD (FedCLG-S even performs slightly worse). This is expected, as there is no variance reduction needed in the IID setting, i.e., $\sigma_g = 0$. Introducing the correction step can bring additional errors caused by the variance of server gradient, resulting in a marginal benefit (or slight weakness) in the IID setting. However, in most realistic scenarios, non-IID data distribution is more common. As seen in Fig. \ref{MNIST_main} (b), under the non-IID setting ($\alpha = 0.2$), even for the MNIST dataset, both FedCLG-S and FedCLG-C outperform the other baselines. Moreover, we observe that in the IID setting, FedAvg converges faster than Server-Only, but in the non-IID setting, the convergence speed of FedAvg decreases significantly. In contrast, for the methods applying additional server training, the convergence speed does not decrease significantly even under the non-IID case. This further validates the necessity of additional server local training, which is consistent with findings in \cite{10001832}.

\begin{figure}[h]
\centering
 \subfloat[MNIST IID]{\includegraphics[width=0.495\linewidth]{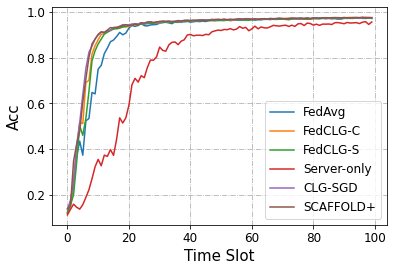}} 
\subfloat[MNIST Non-IID ($\alpha = 0.2$)]{\includegraphics[width=0.495\linewidth]{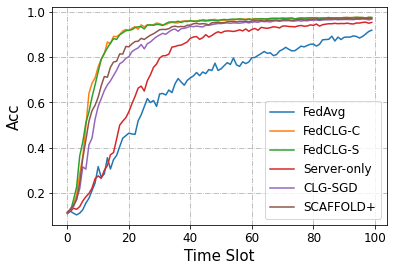}} 
\caption{Convergence performances on MNIST} \label{MNIST_main}
\end{figure}

\textcolor{black}{
We evaluate the proposed methods on CIFAR-10 and CIFAR-100 under non-IID settings. As shown in Figs. \ref{CIFAR_main} and \ref{CIFAR_100}, FedCLG-S and FedCLG-C outperform CLG-SGD by a significant margin. Specifically, in Fig. \ref{CIFAR_main} (a), setting the target test accuracy to 0.5, FedCLG-C requires 134 global rounds, and FedCLG-S requires 144 global rounds. In contrast, CLG-SGD requires 246 global rounds, which is 1.83 (and 1.70) times more than FedCLG-C and FedCLG-S. Furthermore, under high non-IID conditions, FedCLG-C slightly outperforms FedCLG-S, aligning with our theoretical predictions. Despite this, FedCLG-S still significantly surpasses the existing baseline. In scenarios with lower non-IID conditions, both versions of FedCLG exhibit comparable performance. Consequently, the choice between FedCLG-C and FedCLG-S should be based on the conditions of upload/download communication bandwidth.
}
Moreover, although SCAFFOLD+ also applies variance reduction methods, it can only achieve comparable performance with CLG-SGD, indicating that such variance reduction fails to work. We attribute this failure to the use of stale estimated global and local gradients as the guideline to correct the variance, which introduces additional error. The inability to directly apply SCAFFOLD-related methods, which use stale information, further underscores the importance of reasonable exploitation of the server's local data and the necessity of our proposed method, FedCLG.

\begin{figure}[h]
\centering
\subfloat[CIFAR-10 Non-IID ($\alpha = 0.3$)]{\includegraphics[width=0.495\linewidth]{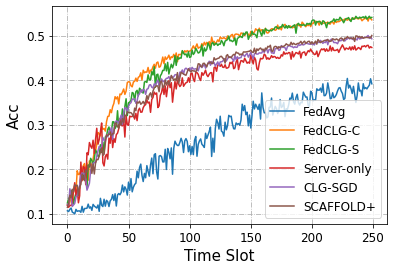}} 
\subfloat[CIFAR-10 Non-IID ($\alpha = 0.8$)]{\includegraphics[width=0.495\linewidth]{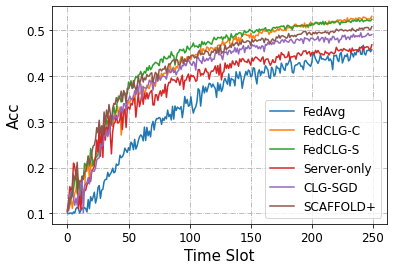}} 
\caption{Convergence performances on CIFAR-10} \label{CIFAR_main}
\end{figure}

\begin{figure}[h]
\centering
\subfloat[CIFAR-100 Non-IID ($\alpha = 0.5$)]{\includegraphics[width=0.495\linewidth]{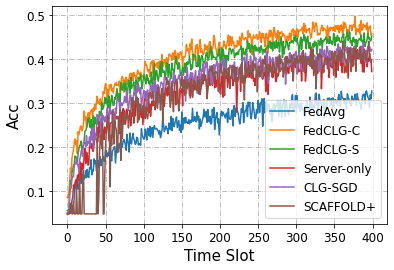}} 
\subfloat[CIFAR-100 Non-IID ($\alpha = 0.8$)]{\includegraphics[width=0.495\linewidth]{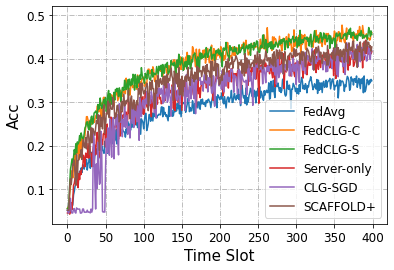}} 
\caption{\textcolor{black}{Convergence performances on CIFAR-100}} \label{CIFAR_100}
\end{figure}

\textbf{Impact of number of participated clients $M$.}
In this series of experiments, we investigate the impact of the number of participated clients $M$ on the convergence speed of our proposed methods and the baseline method, CLG-SGD. We test the experiments on the MNIST dataset while holding all other parameters constant and only varying the number of participated clients $M$. The third column of Table \ref{tab:num_participated} reports the number of global rounds required by each experiment, with the target test accuracy set to $97\%$. The results show that, for both the baseline and our proposed methods, increasing the number of participated clients leads to a decrease in the required number of global rounds and a higher convergence speed. Moreover, under different numbers of participated clients (i.e., 4, 6, 24), the convergence speeds of FedCLG-S and FedCLG-C outperform CLG-SGD. However, with a larger increase in the number of participated clients, the benefit of our proposed methods slightly decreases. For example, with $M = 4$, FedCLG-C achieves a $1.74\times$ speed-up compared to CLG-SGD, while with a larger number of clients ($M = 24$), the speed-up decreases to $1.56\times$. This observation is consistent with our theoretical analysis, which suggests that increasing the number of participated clients in CLG-SGD reduces the error caused by partial participation, thus leading to a smaller benefit from the correction step.

\begin{table}[h]
\caption{Impact of participated clients $M$}
\setlength\extrarowheight{1.5pt}
\label{tab:num_participated}
\vskip 0.15in
\begin{center}
\begin{small}
\begin{sc}
\begin{tabular}{|c|c|c|}
\hline
Methods & Participated Clients & Numbers of round \\
\hline
\hline
\multirow{3}*{CLG-SGD} & $4$ & $68 ~ (1.0 \times)$ \\
\cline{2-3}
& $6$ & $59  ~ (1.0 \times)$ \\
\cline{2-3}
& $24$ & $39  ~ (1.0 \times)$ \\
\hline
\multirow{3}*{FedCLG-S} & $4$ & $42  ~ (1.61 \times)$ \\
\cline{2-3}
& $6$ & $35  ~ (1.68 \times)$ \\
\cline{2-3}
& $24$ & $28  ~ (1.39 \times)$ \\
\hline
\multirow{3}*{FedCLG-C} & $4$ & $39  ~ (1.74 \times)$ \\
\cline{2-3}
& $6$ & $32  ~ (1.84 \times)$ \\
\cline{2-3}
& $24$ & $25  ~ (1.56 \times)$ \\
\hline
\end{tabular}
\end{sc}
\end{small}
\end{center}
\vskip -0.1in
\end{table}

\textbf{Impact of server dataset size.}
To investigate the impact of the server dataset, we first conduct an experiment on MNIST with different levels of $m_s$. In our initial setting, we assume that the server contained $1\%$ of the total training samples each global round. We then extend this value to be $4\%$ and $10\%$. The results, shown in Figure \ref{size_m_s}, indicate that increasing the size of the server dataset can improve the overall convergence speeds of both FedCLG-C and FedCLG-S, as a larger server dataset provides a two-fold improvement. Firstly, more data can be acquired each round, leading to improved server local training. Secondly, the larger server dataset helps us to achieve a more reliable correction step with $g_s$ approaching closer to the actual global optimal direction.

\begin{figure}[h]
\centering
\subfloat[FedCLG-C with different $m_s$]{\includegraphics[width=0.495\linewidth]{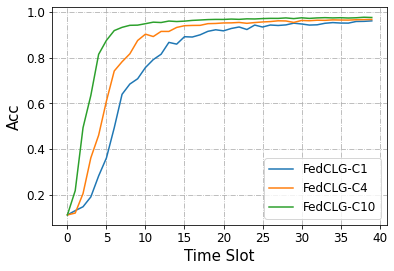}} 
\subfloat[FedCLG-S with different $m_s$]{\includegraphics[width=0.495\linewidth]{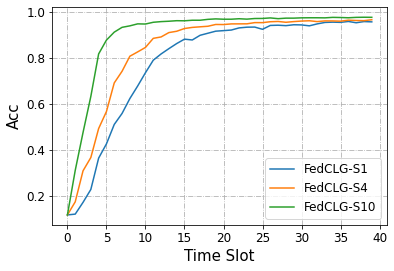}} 
\caption{Impact of server dataset size $m_s$} \label{size_m_s}
\end{figure}

\textcolor{black}{
We then include experiments when the server dataset has a slight distribution shift. This shift is quantified using cosine similarity between the server and overall distributions, specifically targeting scenarios where the shift is small (cosine similarity is approximately 0.95). The results, illustrated in Figure \ref{shift} for both MNIST and CIFAR-10 datasets, reveal that even with this slight distribution shift, our proposed method consistently outperforms baseline approaches. It is important to note that we do not explore scenarios involving significant distribution shifts, as such conditions would effectively reduce the server to another client role, a situation outside the scope of our study in hybrid federated learning.
\begin{figure}[h]
\centering
\subfloat[MNIST (Similarity = 0.946)]{\includegraphics[width=0.495\linewidth]{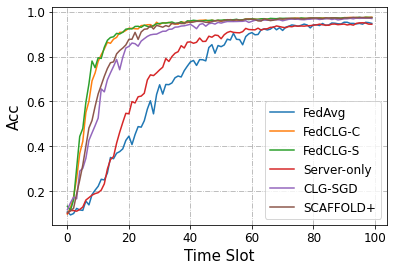}} 
\subfloat[CIFAR-10 (Similarity = 0.952)]{\includegraphics[width=0.495\linewidth]{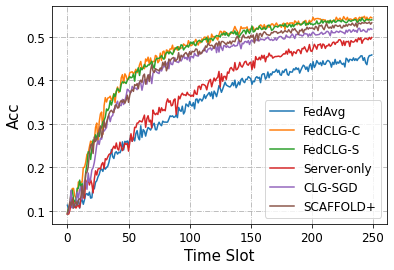}} 
\caption{\textcolor{black}{$m_s$ Distribution Shift}} \label{shift}
\end{figure}
}

\begin{figure}[h]
\centering
\subfloat[FedCLG-C with different $K$]{\includegraphics[width=0.495\linewidth]{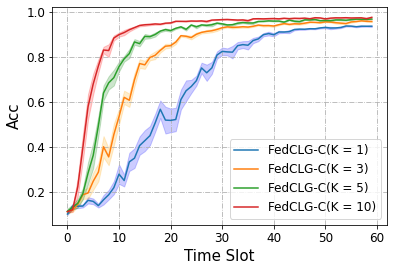}} 
\subfloat[FedCLG-S with different $K$]{\includegraphics[width=0.495\linewidth]{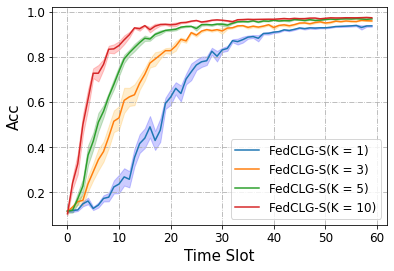}} 
\caption{Impact of client epochs $K$} \label{fig_k}
\end{figure}

\textbf{Impact of client epochs $K$ and server epochs $E$.} In this subsection, we investigate the impact of the number of client epochs $K$ and server epochs $E$. We keep other parameters fixed and only vary the number of client local epochs in Fig. \ref{fig_k}. The convergence results shown in Fig. \ref{fig_k} are consistent with our theoretical analysis, as a larger number of client local epochs $K$ results in increased convergence speed for both FedCLG-C and FedCLG-S. Similarly, in Fig. \ref{fig_e}, we fix all parameters and only change the number of server local epochs. It can be observed that, for both FedCLG-C and FedCLG-S, the largest local epoch ($E=5$) achieves the best convergence performance in both Figs. \ref{fig_e}(a) and (b). 

\begin{figure}[h]
\centering
\subfloat[FedCLG-C with different $E$]{\includegraphics[width=0.495\linewidth]{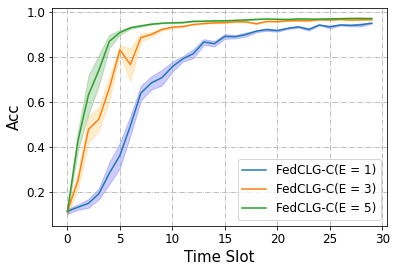}} 
\subfloat[FedCLG-S with different $E$]{\includegraphics[width=0.495\linewidth]{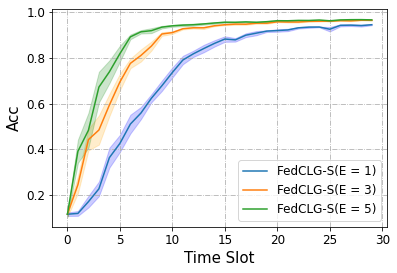}} 
\caption{\textcolor{black}{Impact of server epochs $E$}} \label{fig_e}
\end{figure}

\section{Conclusion}
In this paper, we address the hybrid Federated Learning (FL) setting, where the server has access to a small portion of the total training samples. We consider a more realistic scenario where clients with non-IID data can only partially participate in each server aggregation. Firstly, we provide a novel theoretical analysis for CLG-SGD, the state-of-the-art hybrid FL method. Our analysis reveals the drawbacks of the current method due to clients' partial participation. Motivated by these observations, we propose a novel method called FedCLG, which fully exploits the benefits of a small server dataset. We further study two versions of FedCLG based on different server-client communication scenarios. We provide thorough theoretical analysis and experimental comparisons to validate the proposed methods. Future research will focus on developing a theoretically guaranteed method under an unbounded client-server communication pattern.

\bibliographystyle{IEEEtran}
\bibliography{reference}

\appendices
\section{Proof of Theorem 1}
\label{proof-t1}
\textcolor{black}{
Note that in the following proof, we utilize $g_{t,k}^i$ to represent $\nabla F_i(x_{t,k}^i, \zeta_i)$. 
For each global round $t$, we have the intermediate model after client local training as:
\begin{align}
     x_{t+1}^s = x_{t} + \eta_g \frac{1}{M} \sum_{i\in \mathcal{S}_t}\Delta_t^i,
\end{align}
where $\Delta_t^i = -\sum_{k=0}^{K-1} \eta g_{t,k}^i$. The final model after server local training is as:
\begin{align}
    & x_{t+1}  = x_{t+1}^s - \sum_{e=0}^{E-1}\gamma g_{t+1, e}^s \nonumber\\
    = &   x_{t} - \eta\eta_g \frac{1}{M} \sum_{i\in \mathcal{S}_t}\sum_{k=0}^{K-1} g_{t,k}^i - \gamma\sum_{e=0}^{E-1} \nabla f_s(x_{t+1,e}^s).
\end{align}}

Due to the smoothness in Assumption \ref{assm:smooth}, taking expectation of $f(x_{t+1})$ over the randomness at communication round $t$, we have:
\begin{align}
    \mathbb{E}_t [f(x_{t+1})] &\leq f(x_t) + \underbrace{\bigg< \nabla f(x_t), \mathbb{E}_t [x_{t+1} - x_t] \bigg>}_{T_1}\nonumber\\
    & + \underbrace{\frac{L}{2} \mathbb{E}_t [\| x_{t+1} - x_t \|^2]}_{T_2}. \label{ineq_smooth}
\end{align}

\textcolor{black}{
We first bound $T_1$ as follows:
\begin{align}
    & T_1 = \bigg< \nabla f(x_t), \mathbb{E}_t [x_{t+1} - x_t] \bigg>\nonumber\\
    & = -\eta\eta_g\bigg< \nabla f(x_t), \mathbb{E}_t [\frac{1}{M} \sum_{i\in \mathcal{S}_t}\sum_{k=0}^{K-1} g_{t,k}^i]\bigg> \nonumber\\
    & - \gamma\bigg<\nabla f(x_t), \mathbb{E}_t[\sum_{e=0}^{E-1} \nabla f_s(x_{t+1,e}^s)] \bigg>\nonumber\\
    & = \underbrace{-\eta\eta_g\bigg< \nabla f(x_t), \mathbb{E}_t [\frac{1}{M} \sum_{i\in \mathcal{S}_t}\sum_{k=0}^{K-1} \nabla f_i(x_{t,k}^i)]\bigg>}_{T_3}\nonumber\\
    & \underbrace{- \gamma\bigg<\nabla f(x_t), \mathbb{E}_t[\sum_{e=0}^{E-1} \nabla f_s(x_{t+1,e}^s)] \bigg>}_{T_4},
\end{align}There the second equality is due to the assumption of an unbiased local gradient estimate.}
The term $T_3$ can be bounded as follows:
\begin{align}
    & T_3 = -\eta\eta_g\bigg< \nabla f(x_t), \mathbb{E}_t [\frac{1}{M} \sum_{i\in \mathcal{S}_t}\sum_{k=0}^{K-1} \nabla f_i(x_{t,k}^i)]\bigg> \nonumber\\
     & = -\frac{\eta\eta_g}{K}\bigg< K\nabla f(x_t), \mathbb{E}_t [\frac{1}{N} \sum_{i\in [N]}\sum_{k=0}^{K-1} \nabla f_i(x_{t,k}^i)]\bigg> \nonumber\\
     & = \underbrace{\frac{\eta\eta_g}{2K}\mathbb{E}_t [\|\frac{1}{N} \sum_{i\in [N]}\sum_{k=0}^{K-1} \nabla f_i(x_{t,k}^i) - K\nabla f(x_{t})\|^2]}_{T_5} \nonumber\\
     & - \frac{\eta\eta_g K}{2}\|\nabla f(x_{t})\|^2 - \frac{\eta\eta_g}{2K}\mathbb{E}_t [\|\frac{1}{N} \sum_{i\in [N]}\sum_{k=0}^{K-1} \nabla f_i(x_{t,k}^i)\|^2]. \label{ineq_t3}
\end{align}
The last equality is due to the fact that  $\bigg<x, y \bigg> = \frac{1}{2} [ \| x \|^2 + \| y \|^2 - \| x - y \|^2 ]$. Then we can bound $T_5$ as:
\begin{align}
    &T_5 = \frac{\eta\eta_g}{2K}\mathbb{E}_t [\|\frac{1}{N} \sum_{i\in [N]}\sum_{k=0}^{K-1} \nabla f_i(x_{t,k}^i) - K\nabla f(x_{t})\|^2] \nonumber \\
    & = \frac{\eta\eta_g}{2K}\mathbb{E}_t [\|\frac{1}{N} \sum_{i\in [N]}\sum_{k=0}^{K-1} [\nabla f_i(x_{t,k}^i) - \nabla f_i(x_{t})]\|^2] \nonumber \\
    & \leq \frac{\eta\eta_g}{2N}\sum_{i\in [N]}\sum_{k=0}^{K-1} \mathbb{E}_t [\| [\nabla f_i(x_{t,k}^i) - \nabla f_i(x_{t})]\|^2] \nonumber \\
    & \leq \underbrace{\frac{\eta\eta_g L^2}{2N}\sum_{i\in [N]}\sum_{k=0}^{K-1} \mathbb{E}_t [\| x_{t,k}^i - x_{t}\|^2] }_{T_6}
\end{align}
The first inequality is based on Cauchy-Schwarz inequality. \textcolor{black}{Then we have $T_6$ be bounded as:
\begin{align}
    & T_6 =  \frac{\eta\eta_g L^2}{2N}\sum_{i\in [N]}\sum_{k=0}^{K-1} \mathbb{E}_t [\| \eta \sum_{\tau = 0}^k g_{t, \tau}^i\|^2] \nonumber\\
    & = \frac{\eta\eta_g L^2}{2N}\sum_{i\in [N]}\sum_{k=0}^{K-1} \mathbb{E}_t [\| \eta \sum_{\tau = 0}^k (g_{t, \tau}^i  - \nabla f_i(x_{t, k}^i) )\|^2] \nonumber\\
    & + \frac{\eta\eta_g L^2}{2N}\sum_{i\in [N]}\sum_{k=0}^{K-1} \mathbb{E}_t [\| \eta \sum_{\tau = 0}^k \nabla f_i(x_{t, k}^i) \|^2]
    \nonumber\\
    & = \frac{\eta\eta_g L^2}{2N}\sum_{i\in [N]}\sum_{k=0}^{K-1} \mathbb{E}_t [\| \eta \sum_{\tau = 0}^k (g_{t, \tau}^i - \nabla f_i(x_{t, k}^i) )\|^2] \nonumber\\
    &  + \frac{\eta\eta_g L^2 K}{2N}\sum_{i\in [N]}\sum_{k=0}^{K-1} \mathbb{E}_t [\sum_{\tau = 0}^k \| \eta \nabla f_i(x_{t, k}^i) \|^2]
    \nonumber\\
    & \leq \frac{\eta^3\eta_g L^2 K^2}{2} \sum_{k=0}^{K-1} \underbrace{\frac{1}{N}\sum_{i\in [N]}\mathbb{E}_t [\|\nabla f_i(x_{t, k}^i)\|^2]}_{T_7} + \frac{\eta^3\eta_g L^2 K^2 \sigma_l^2}{2},
    \label{ineq_t6}
\end{align} where the second equality is based on the assumption \ref{a_variance}.
}
To further bound $T_7$, we have:
\begin{align}
    & T_7 \leq \frac{3}{N}\sum_{i\in [N]}\mathbb{E}_t [\|\nabla f_i(x_{t, k}^i)-\nabla f_i(x_{t})\|^2] \nonumber\\
    & + \frac{3}{N}\sum_{i\in [N]}\mathbb{E}_t [\|\nabla f_i(x_{t}) - \nabla f(x_{t})\|^2] + \frac{3}{N}\sum_{i\in [N]}\mathbb{E}_t [\|\nabla f(x_{t})\|^2] \nonumber \\
    & \leq \frac{3 L^2}{N} \sum_{i\in [N]} \mathbb{E}_t [\| x_{t} - x_{t, k}^i\|^2] + 3 \sigma_g^2 + 3\mathbb{E}_t [\|\nabla f(x_{t})\|^2], 
\end{align}
where the last inequality is due to the assumptions \ref{assm:smooth},\ref{a_variance}. \textcolor{black}{Substituting $T_7$ to (\ref{ineq_t6}), we have:
\begin{align}
    & T_6 \leq \frac{\eta^3\eta_g L^2 K^2}{2N} \sum_{i\in [N]}\sum_{k=0}^{K-1}( 3 L^2\mathbb{E}_t [\| x_{t} - x_{t, k}^i\|^2] \nonumber\\
    & + 3 \sigma_g^2 + 3\mathbb{E}_t [\|\nabla f(x_{t})\|^2]) + \frac{\eta^3\eta_g L^2 K^2 \sigma_l^2}{2}\nonumber\\
    & \leq \frac{3\eta^3\eta_g L^2 K^3}{2 (1-\mathcal{B})}\sigma_g^2 + \frac{\eta^3\eta_g L^2 K^2 \sigma_l^2}{2 (1-\mathcal{B})} \nonumber\\
    &  + \frac{3\eta^3\eta_g L^2 K^3}{2 (1-\mathcal{B})}\mathbb{E}_t [\|\nabla f(x_{t})\|^2] \nonumber\\
    & \leq \frac{3\eta^3\eta_g L^2 K^2 (\sigma_l^2+3K\sigma_g^2)}{4} + \frac{\eta\eta_g K}{4}\mathbb{E}_t [\|\nabla f(x_{t})\|^2]
\end{align}
where $\mathcal{B} = 3\eta^2L^2K^2$ and let $\eta \leq \frac{1}{3LK}$ such that $\mathcal{B} \leq \frac{1}{3}$, $\frac{1}{1 - \mathcal{B}} \leq \frac{3}{2}$ and $\frac{\mathcal{B}}{1 - \mathcal{B}} \leq \frac{1}{2}$. Then we substitute it to (\ref{ineq_t3}), we can get:
\begin{align}
    &T_3 \leq \frac{3\eta^3\eta_g L^2 K^2 (\sigma_l^2+3K\sigma_g^2)}{4}  - \frac{\eta\eta_g K}{4}\|\nabla f(x_{t})\|^2 \nonumber\\
    &- \frac{\eta\eta_g}{2K}\mathbb{E}_t [\|\frac{1}{N} \sum_{i\in [N]}\sum_{k=0}^{K-1} \nabla f_i(x_{t,k}^i)\|^2].
\end{align}}

Next, the term $T_4$ can be bounded as 
\begin{align}
    & T_4 = - \frac{\gamma}{E}\bigg<E\nabla f(x_t), \mathbb{E}_t[\sum_{e=0}^{E-1} \nabla f(x_{t+1,e}^s)] \bigg> \nonumber\\
    & = \underbrace{\frac{\gamma}{2E} \mathbb{E}_t[\|\sum_{e=0}^{E-1} \nabla f(x_{t+1,e}^s) - E\nabla f(x_t)\|^2]}_{T_8} \nonumber\\
    &- \frac{\gamma E}{2}\|\nabla f(x_{t})\|^2 - \frac{\gamma}{2E}\mathbb{E}_t[\|\sum_{e=0}^{E-1} \nabla f(x_{t+1,e}^s)\|^2]. \label{ineq_t4}
\end{align}
The last equality is due to the fact that  $\bigg<x, y \bigg> = \frac{1}{2} [ \| x \|^2 + \| y \|^2 - \| x - y \|^2 ]$. \textcolor{black}{Then the term $T_8$ can be bounded as follows:
\begin{align}
    & T_8 = \frac{\gamma}{2E} \mathbb{E}_t[\|\sum_{e=0}^{E-1} \nabla f(x_{t+1,e}^s) - E\nabla f(x_t)\|^2] \nonumber\\
    & \leq \frac{\gamma}{2}\sum_{e=0}^{E-1} \mathbb{E}_t[\| \nabla f(x_{t+1,e}^s) - \nabla f(x_t)\|^2] \nonumber\\
    & \leq \underbrace{\frac{\gamma L^2}{2}\sum_{e=0}^{E-1} \mathbb{E}_t[\| x_{t+1,e}^s - x_t\|^2]}_{T_9} \nonumber\\
    & \leq {\gamma L^2}\sum_{e=0}^{E-1} \mathbb{E}_t[\| x_{t+1,e}^s - x_{t+1}^s\|^2]  + {\gamma L^2}\sum_{e=0}^{E-1} \mathbb{E}_t[\| x_{t+1}^s - x_t\|^2]\nonumber\\
    & \leq \underbrace{{\gamma L^2}\sum_{e=0}^{E-1} \mathbb{E}_t[\| x_{t+1,e}^s - x_{t+1}^s\|^2]}_{T_{10}} + \frac{\eta^2 \eta_g^2\gamma E L^2 K \sigma_l^2}{M}\nonumber\\
    & + \eta^2 \eta_g^2\gamma E L^2 \underbrace{\mathbb{E}_t [\|\frac{1}{M} \sum_{i\in \mathcal{S}_t}\sum_{k=0}^{K-1} \nabla f_i(x_{t,k}^i) \|^2]}_{T_{11}}.
\end{align}}
Then bounding $T_{10}$, we have:
\begin{align}
    & T_{10} = \gamma L^2\sum_{e=0}^{E-1}\mathbb{E}_t[\| \sum_{\tau_e=0}^{e-1} \gamma \nabla f_s(x_{t+1, \tau_e}^s)\|^2]\nonumber\\
    & \leq \gamma L^2\sum_{e=0}^{E-1}\gamma^2\mathbb{E}_t[\| \sum_{\tau_e=0}^{e-1} \nabla f_s(x_{t+1, e}^s)\|^2]\nonumber\\
    &  \leq 3\gamma^3L^2\sum_{e=0}^{E-1}\mathbb{E}_t[\|\sum_{\tau_e=0}^{e-1} ( \nabla f_s(x_{t+1, e}^s) -  \nabla f(x_{t+1, e}^s))\|^2]\nonumber\\
    & + 3\gamma^3 L^2\sum_{e=0}^{E-1}\mathbb{E}_t[\| \sum_{\tau_e=0}^{e-1} (\nabla f(x_{t+1, e}^s) -  \nabla f(x_{t}))\|^2] \nonumber\\
    & + 3\gamma^3 L^2\sum_{e=0}^{E-1}\mathbb{E}_t[\|\sum_{\tau_e=0}^{e-1}  \nabla f(x_{t})\|^2] \nonumber\\
    & \leq \frac{3\gamma^3 E^2 L^2 \sigma^2}{m_s} + 3\gamma^3E^2 L^4\sum_{e=0}^{E-1}\mathbb{E}_t[\|x_{t+1, e}^s -  x_{t}\|^2] \nonumber\\
    & + 3\gamma^3E^2 L^2\sum_{e=0}^{E-1}\mathbb{E}_t[\| \nabla f(x_{t})\|^2],
\end{align}
where the first term in the third inequality is due to the fact that $\mathbb{E} [\|x_1 + \cdots + x_n \|^2] = \mathbb{E}[\|x_1\|^2 + \cdots + \| x_n \|^2]$ if $x_i$ is independent with zero mean and assumption \ref{a_variance}, 

Then to bound the term $T_{11}$, we have:
\begin{align}
    &T_{11} \leq 3\mathbb{E}_t [ \|\frac{1}{M} \sum_{i\in \mathcal{S}_t}\sum_{k=0}^{K-1}[ \nabla f_i(x_{t,k}^i) - \nabla f_i(x_{t}) ]\|^2]\nonumber\\
    & + 3\mathbb{E}_t [ \|\frac{1}{M} \sum_{i\in \mathcal{S}_t}\sum_{k=0}^{K-1}[\nabla f_i(x_{t}) -  \nabla f(x_{t}) ]\|^2] \nonumber\\
    & + 3\mathbb{E}_t [ \|\frac{1}{M} \sum_{i\in \mathcal{S}_t}\sum_{k=0}^{K-1} \nabla f(x_{t}) \|^2] \nonumber\\
    & \leq 3\mathbb{E}_t [ \frac{1}{M} \sum_{i\in \mathcal{S}_t}\|\sum_{k=0}^{K-1}[ \nabla f_i(x_{t,k}^i) - \nabla f_i(x_{t}) ]\|^2]\nonumber\\
    & + 3\mathbb{E}_t [ \|\frac{1}{M} \sum_{i\in \mathcal{S}_t}\sum_{k=0}^{K-1}[\nabla f_i(x_{t}) -  \nabla f(x_{t}) ]\|^2] \nonumber \\
    & + 3 K^2\mathbb{E}_t [ \|\nabla f(x_{t}) \|^2] \nonumber\\
    & \leq \frac{3}{N} \sum_{i\in [N]}\mathbb{E}_t [\|\sum_{k=0}^{K-1}[ \nabla f_i(x_{t,k}^i) - \nabla f_i(x_{t}) ]\|^2]\nonumber\\
    & + 3\underbrace{\mathbb{E}_t [ \|\frac{1}{M} \sum_{i\in \mathcal{S}_t}\sum_{k=0}^{K-1}[\nabla f_i(x_{t}) -  \nabla f(x_{t}) ]\|^2]}_{T_{12}} \nonumber\\
    & + 3 K^2\mathbb{E}_t [ \|\nabla f(x_{t}) \|^2],  \label{ineq_t11}
\end{align}
where the last inequality is due to the server's uniformly selection without replacement. Next we need to bound $T_{12}$. For convenience, we utilize $\delta_t^i = \sum_{k=0}^{K-1} \nabla f_i(x_t)$ and  $\delta_t = \sum_{k=0}^{K-1} \nabla f(x_t)$ in the following step.
\begin{align}
    &T_{12} = \mathbb{E}_t [\|\frac{1}{M} \sum_{i\in \mathcal{S}_t}\delta_t^i - \delta_t  \|^2] \nonumber\\
    & = \frac{1}{M^2}\mathbb{E}_t [\| \sum_{i\in [N]}(\mathcal{I}(i\in \mathcal{S}_t))(\delta_t^i - \delta_t) \|^2] \nonumber\\
    & = \frac{1}{M^2}\mathbb{E}_t [ \sum_{i\in [N]}(\mathcal{I}(i\in \mathcal{S}_t))^2 \|\delta_t^i - \delta_t\|^2] \nonumber\\
    & + \frac{1}{M^2}\mathbb{E}_t[\sum_{i\in [N]}\sum_{j \ne i\in [N]} \mathcal{I}(i\in \mathcal{S}_t)\mathcal{I}(j\in \mathcal{S}_t)\bigg< \delta_t^i - \delta_t, \delta_t^j - \delta_t\bigg>] \nonumber\\
    & = \frac{1}{M^2} \frac{M}{N}\mathbb{E}_t [ \sum_{i\in [N]}\|\delta_t^i - \delta_t\|^2] \nonumber\\
    & + \frac{M(M-1)}{N(N-1)}\frac{1}{M^2}\mathbb{E}_t[\|\sum_{i \in [N]} (\delta_t^i - \delta_t)\|^2] \nonumber\\
    & - \frac{M(M-1)}{N(N-1)}\frac{1}{M^2}\mathbb{E}_t [ \sum_{i\in [N]}\|\delta_t^i - \delta_t\|^2] \nonumber\\
    & = \frac{N-M}{MN(N-1)} \sum_{i\in [N]}\mathbb{E}_t [ \|\delta_t^i - \delta_t\|^2] \nonumber\\
    & \leq \frac{(N-M)K}{MN(N-1)} \sum_{i\in [N]}\sum_{k =0}^{K-1}\mathbb{E}_t [ \|\nabla f_i(x_t) - \nabla f(x_t)\|^2] \nonumber\\
    & \leq \frac{(N-M)K^2}{M(N-1)} \sigma_g^2,
\end{align}
where the second equality is due to the server's uniformly selection without replacement and third equality is due to $\sum_{i \in [N]} (\delta_t^i - \delta_t) = 0$. \textcolor{black}{Then substituting the result to (\ref{ineq_t11}), we have:
\begin{align}
    &T_{11} \leq \frac{3}{N} \sum_{i\in [N]}\mathbb{E}_t [\|\sum_{k=0}^{K-1}[ \nabla f_i(x_{t,k}^i) - \nabla f_i(x_{t}) ]\|^2]\nonumber\\
    & + 3\frac{(N-M)K^2}{M(N-1)} \sigma_g^2  + 3 K^2\mathbb{E}_t [ \|\nabla f(x_{t}) \|^2] \nonumber\\
    & \leq  \frac{9\eta^2 L^2 K^3 (\sigma_l^2+3K\sigma_g^2) }{2} + \frac{3K^2}{2}\mathbb{E}_t [\|\nabla f(x_{t})\|^2] \nonumber\\
    & + 3\frac{(N-M)K^2}{M(N-1)} \sigma_g^2  + 3 K^2\mathbb{E}_t [ \|\nabla f(x_{t}) \|^2] \nonumber\\
    & = \frac{9\eta^2 L^2 K^3 (\sigma_l^2+3K\sigma_g^2) }{2} \nonumber\\
    & + \frac{9K^2}{2}\mathbb{E}_t [\|\nabla f(x_{t})\|^2] + 3\frac{(N-M)K^2}{M(N-1)} \sigma_g^2.
\end{align}}
\textcolor{black}{
Substituting the results of $T_{10}$ and $T_{11}$, we have:
\begin{align}
    &T_{9} \leq  \frac{3\gamma^3 E^2 L^2 \sigma^2}{m_s} + 3\gamma^3E^2 L^4\sum_{e=0}^{E-1}\mathbb{E}_t[\|x_{t+1, e}^s -  x_{t}\|^2] \nonumber\\
    & + 3\gamma^3E^3 L^2\mathbb{E}_t[\| \nabla f(x_{t})\|^2] + \eta^2 \eta_g^2\gamma E L^2 \frac{9\eta^2 L^2 K^3 (\sigma_l^2+3K\sigma_g^2) }{2} \nonumber\\
    & + (\eta^2 \eta_g^2\gamma E L^2)\frac{9K^2}{2}\mathbb{E}_t [\|\nabla f(x_{t})\|^2] \nonumber\\
    &+ 3(\eta^2 \eta_g^2\gamma E L^2)\frac{(N-M)K^2}{M(N-1)} \sigma_g^2 + \frac{\eta^2 \eta_g^2\gamma E L^2 K \sigma_l^2}{M} \nonumber\\
    & \leq  \frac{3\gamma^3 E^2 L^2 \sigma^2}{m_s} + 3\gamma^3E^2 L^4\sum_{e=0}^{E-1}\mathbb{E}_t[\|x_{t+1, e}^s -  x_{t}\|^2] + \frac{\eta^2 \eta_g^2\gamma E L^2 K \sigma_l^2}{M}\nonumber\\
    & + 3\gamma^3E^3 L^2\mathbb{E}_t[\| \nabla f(x_{t})\|^2] + \frac{3\eta^4\eta_g^2 L^3 K^3(\sigma_l^2 + 3K\sigma_g^2) }{4} \nonumber\\
    & + \frac{3\eta^2 \eta_g^2LK^2}{4}\mathbb{E}_t [\|\nabla f(x_{t})\|^2] + \frac{\eta^2 \eta_g^2L(N-M)K^2}{2M(N-1)} \sigma_g^2 \nonumber\\
    & \leq  \frac{3\gamma^3 E^2 L^2 \sigma^2}{(1-\mathcal{A})m_s} + \frac{\mathcal{A}\gamma E}{2(1-\mathcal{A})} \mathbb{E}_t[\| \nabla f(x_{t})\|^2] \nonumber\\
    & + \frac{3\eta^4\eta_g^2 L^3 K^3(\sigma_l^2 + 3K\sigma_g^2)  }{4(1-\mathcal{A})} + \frac{\eta^2 \eta_g^2\gamma E L^2 K \sigma_l^2}{M(1- \mathcal{A})}\nonumber\\
    & + \frac{3\eta^2 \eta_g^2LK^2}{4(1-\mathcal{A})}\mathbb{E}_t [\|\nabla f(x_{t})\|^2] + \frac{\eta^2 \eta_g^2L(N-M)K^2}{2M(N-1)(1-\mathcal{A})} \sigma_g^2 \nonumber\\
    & \leq  \frac{18\gamma^3 E^2 L^2 \sigma^2}{5m_s} + \frac{\gamma E}{10} \mathbb{E}_t[\| \nabla f(x_{t})\|^2] \nonumber\\
    & + \frac{9\eta^4\eta_g^2 L^3 K^3(\sigma_l^2+3K\sigma_g^2)}{10} + \frac{9\eta^2 \eta_g^2LK^2}{10}\mathbb{E}_t [\|\nabla f(x_{t})\|^2] \nonumber\\
    & + \frac{3\eta^2 \eta_g^2L(N-M)K^2}{5M(N-1)} \sigma_g^2  + \frac{\eta^2 \eta_g^2 L K \sigma_l^2}{5M},
\end{align}
where $\mathcal{A} = 6\gamma^2L^2E^2$ and let $\gamma \leq \frac{1}{6LE}$ such that $\mathcal{A} \leq \frac{1}{6}$, $\frac{1}{1 - \mathcal{A}} \leq \frac{6}{5}$ and $\frac{\mathcal{A}}{1 - \mathcal{A}} \leq \frac{1}{5}$. Substituting the above result to (\ref{ineq_t4}), we have:
\begin{align}
    & T_4 \leq  \frac{18\gamma^3 E^2 L^2 \sigma^2}{5m_s} - \frac{2\gamma E}{5} \mathbb{E}_t[\| \nabla f(x_{t})\|^2] \nonumber\\
    &+ \frac{9\eta^4\eta_g^2 L^3 K^3 (\sigma_l^2+3K\sigma_g^2)}{10}  + \frac{9\eta^2 \eta_g^2LK^2}{10}\mathbb{E}_t [\|\nabla f(x_{t})\|^2] \nonumber\\
    &+ \frac{3\eta^2 \eta_g^2L(N-M)K^2}{5M(N-1)} \sigma_g^2 - \frac{\gamma}{2E}\mathbb{E}_t[\|\sum_{e=0}^{E-1} \nabla f(x_{t+1,e}^s)\|^2]\nonumber\\
    &+ \frac{\eta^2 \eta_g^2 L K \sigma_l^2}{5M}.
\end{align}
Combing the inequalities of $T_3$ and $T_4$, we can then bound $T_1$ as:
\begin{align}
    &T_1 \leq \frac{3\eta^3\eta_g L^2 K^2 (\sigma_l^2+3K\sigma_g^2)}{4} - (\frac{2\gamma E}{5}+\frac{\eta\eta_g K}{4})\|\nabla f(x_{t})\|^2 \nonumber\\
    & - \frac{\eta\eta_g}{2K}\mathbb{E}_t [\|\frac{1}{N} \sum_{i\in [N]}\sum_{k=0}^{K-1} \nabla f_i(x_{t,k}^i)\|^2]  + \frac{18\gamma^3 E^2 L^2 \sigma^2}{5m_s} \nonumber\\
    & + \frac{9\eta^4\eta_g^2 L^3 K^3(\sigma_l^2+3K\sigma_g^2) }{10} + \frac{9\eta^2 \eta_g^2LK^2}{10}\mathbb{E}_t [\|\nabla f(x_{t})\|^2]\nonumber\\
    &  + \frac{3\eta^2 \eta_g^2L(N-M)K^2}{5M(N-1)} \sigma_g^2 - \frac{\gamma}{2E}\mathbb{E}_t[\|\sum_{e=0}^{E-1} \nabla f(x_{t+1,e}^s)\|^2]\nonumber\\
    &  + \frac{\eta^2 \eta_g^2L K \sigma_l^2}{5M}.
\end{align}
The rest term $T_2$ can be bounded as:
\begin{align}
    & T_2 \leq \eta^2\eta_g^2L\underbrace{\mathbb{E}_t [\|\frac{1}{M} \sum_{i\in \mathcal{S}_t}\sum_{k=0}^{K-1} \nabla f_i(x_{t,k}^i) \|^2] }_{T_{11}} + \frac{\eta^2\eta_g^2LK \sigma_l^2}{M}\nonumber\\
    & + \gamma^2L\underbrace{\mathbb{E}_t [\|\sum_{e=0}^{E-1} \nabla f_s(x_{t+1,e}^s) \|^2]}_{T_{13}}. \label{ineq_t2}
\end{align}}

The only rest term is $T_{13}$, which can be bounded as:
\begin{align}
    &T_{13} = \mathbb{E}_t [\|\sum_{e=0}^{E-1} (\nabla f_s(x_{t+1,e}^s) - \nabla f(x_{t+1,e}^s) + \nabla f(x_{t+1,e}^s))\|^2] \nonumber\\
    & = \mathbb{E}_t [\|\sum_{e=0}^{E-1} (\nabla f_s(x_{t+1,e}^s) - \nabla f(x_{t+1,e}^s))\|^2]  \nonumber\\
    & + \mathbb{E}_t [\|\sum_{e=0}^{E-1} \nabla f(x_{t+1,e}^s)\|^2]  \leq \frac{ E \sigma^2}{m_s} + \mathbb{E}_t [\|\sum_{e=0}^{E-1} \nabla f(x_{t+1,e}^s)\|^2],
\end{align}
where the third equality is due to the fact that  $\mathbb{E}[\| x \|^2] = \mathbb{E}[\| x - \mathbb{E}[x] \|^2] + \| \mathbb{E}[x] \|^2$ and the last inequality is due to assumption \ref{assm:unbiased-local} and  $\mathbb{E} [\|x_1 + \cdots + x_n \|^2] \leq n \mathbb{E}[\|x_1\|^2 + \cdots + \| x_n \|^2]$.
\textcolor{black}{
Substituting the result of $T_{11}$ and $T_{13}$, we can finally bound $T_{2}$ as:
\begin{align}
    &T_2 \leq \frac{9\eta^4 \eta_g^2 L^3 K^3(\sigma_l^2+3K\sigma_g^2) }{2} + \frac{9K^2\eta^2\eta_g^2L}{2}\mathbb{E}_t [\|\nabla f(x_{t})\|^2] \nonumber\\
    & + 3\frac{(N-M)K^2\eta^2\eta_g^2L}{M(N-1)} \sigma_g^2 + \frac{ \gamma^2L E \sigma^2}{m_s} \nonumber\\
    & + \gamma^2L\mathbb{E}_t [\|\sum_{e=0}^{E-1} \nabla f(x_{t+1,e}^s)\|^2] + \frac{\eta^2\eta_g^2LK \sigma_l^2}{M}.
\end{align}}

\textcolor{black}{
With both $T_1$ and $T_2$ bounded, we finally have:
\begin{align}
     & \mathbb{E}_t [f(x_{t+1})] \leq f(x_t) - (\frac{2\gamma E}{5}+\frac{\eta\eta_g K}{20})\|\nabla f(x_{t})\|^2 + \frac{8\gamma^2 E L \sigma^2}{5m_s}\nonumber\\
     & + \frac{57\eta^3\eta_g L^2 K^3\sigma_g^2}{20} + \frac{18\eta^2 \eta_g^2L(N-M)K^2}{5M(N-1)} \sigma_g^2 \nonumber\\
     & + \frac{19\eta^3\eta_g L^2 K^2\sigma_l^2}{20} + \frac{6\eta^2\eta_g^2LK \sigma_l^2}{5M},
\end{align}
where $\gamma \leq \frac{1}{6EL}$, $\eta \leq \frac{1}{3KL}$ and $\eta \eta_g \leq \frac{1}{27KL}$.}
\textcolor{black}{
Rearranging and summing from $t = 0, \dots, T-1$, we have the convergence as:
\begin{align}
    & \min_{t \in [T]} \mathbb{E}\|\nabla f(x_t)\|_2^2 = \mathcal{O}\bigg(\frac{(f_0 - f_*)}{T(\gamma E + \eta\eta_g K)}\bigg) + \mathcal{O}\bigg(\frac{\eta^3\eta_g L^2 K^3\sigma_g^2}{\gamma E + \eta\eta_g K}\bigg)  \nonumber\\
    &  + \mathcal{O}\bigg(\frac{\gamma^2 E L \sigma^2}{m_s(\gamma E + \eta\eta_g K)} \bigg) +  \mathcal{O}\bigg(\frac{(N-M)K^2\eta^2\eta_g^2L\sigma_g^2}{M(N-1)(\gamma E + \eta\eta_g K)}\bigg) \nonumber\\
    &  + \underbrace{\mathcal{O}\bigg(\frac{\eta^3\eta_g L^2 K^2\sigma_l^2}{\gamma E + \eta\eta_g K}\bigg) +\mathcal{O}\bigg(\frac{\eta^2\eta_g^2LK \sigma_l^2}{M(\gamma E + \eta\eta_g K)}\bigg)}_{\text{stochastic gradient error}} ,
\end{align}
where $f_0 = f(x_0)$, $f_* = f(x_*)$.}

Now we finish the proof of theorem \ref{th1}.

\section{Proof of Theorem 2}
\label{proof-t2}
\textcolor{black}{
After each global round $t$, we have the new model $x_{t+1}$ as:
\begin{align}
    & x_{t+1}  =  x_{t} - \eta\eta_g \frac{1}{M} \sum_{i\in \mathcal{S}_t}\sum_{k=0}^{K-1} (g_{t,k}^i + \nabla f_s(x_t) - g_t^i)   \nonumber\\
    & - \gamma\sum_{e=0}^{E-1} \nabla f_s(x_{t+1,e}^s).
\end{align}
Similar to the proof of 1, due to the smoothness in Assumption 1, taking the expectation of $f(x_{t+1})$ over the randomness at communication round $t$, we have:
\begin{align}
    \mathbb{E}_t [f(x_{t+1})] &\leq f(x_t) + \underbrace{\bigg< \nabla f(x_t), \mathbb{E}_t [x_{t+1} - x_t] \bigg>}_{C_1}\nonumber\\
    & + \underbrace{\frac{L}{2} \mathbb{E}_t [\| x_{t+1} - x_t \|^2]}_{C_2}. \label{ineq_smooth_c}
\end{align}
Focusing on $\mathbb{E}_t [x_{t+1} - x_t]$ in the term $C_1$, we can find that:
\begin{align}
    & \mathbb{E}_t [x_{t+1} - x_t]\nonumber\\
    & = -\mathbb{E}_t [\eta\eta_g \frac{1}{M} \sum_{i\in \mathcal{S}_t}\sum_{k=0}^{K-1} \nabla f_i(x_{t,k}^i) + \gamma\sum_{e=0}^{E-1} \nabla f_s(x_{t+1,e}^s)] \nonumber\\
    & - \underbrace{\mathbb{E}_t[\eta\eta_g \frac{1}{M} \sum_{i\in \mathcal{S}_t}\sum_{k=0}^{K-1}(\nabla f_s(x_t) - g_{t}^i)]}_{=0}
\end{align}
In the following proof steps, some of the steps are similar to the previous proof of 1. We will skip the intermediate steps here. We can bound $C_1$ similar to $T_1$ as:
\begin{align}
    & C_1 = \underbrace{-\eta\eta_g\bigg< \nabla f(x_t), \mathbb{E}_t [\frac{1}{M} \sum_{i\in \mathcal{S}_t}\sum_{k=0}^{K-1} \nabla f_i(x_{t,k}^i)]\bigg>}_{C_3}\nonumber\\
    & \underbrace{- \gamma\bigg<\nabla f(x_t), \mathbb{E}_t[\sum_{e=0}^{E-1} \nabla f_s(x_{t+1,e}^s)] \bigg>}_{C_4}.
\end{align}
The term $C_3$ can be bounded as follows:
\begin{align}
    & C_3 = \underbrace{\frac{\eta\eta_g}{2K}\mathbb{E}_t [\|\frac{1}{N} \sum_{i\in [N]}\sum_{k=0}^{K-1} \nabla f_i(x_{t,k}^i) - K\nabla f(x_{t})\|^2]}_{C_5} \nonumber\\
     & - \frac{\eta\eta_g K}{2}\|\nabla f(x_{t})\|^2 - \frac{\eta\eta_g}{2K}\mathbb{E}_t [\|\frac{1}{N} \sum_{i\in [N]}\sum_{k=0}^{K-1} \nabla f_i(x_{t,k}^i)\|^2]. \label{ineq_c3}
\end{align}
}
Then we can bound $C_5$ as:
\begin{align}
    &C_5 \leq \frac{\eta\eta_g L^2}{2N}\sum_{i\in [N]}\sum_{k=0}^{K-1} \underbrace{\mathbb{E}_t [\| x_{t,k}^i - x_{t}\|^2] }_{C_6}
\end{align}
\textcolor{black}{
Next, consider $\eta \leq \frac{1}{8LK}$, the term $C_6$ can be bounded as:
\begin{align}
     &C_6  = \mathbb{E}_t [\|x_{t,k-1}^i - \eta(g_{t, k-1}^i - \nabla f_i (x_{t,k-1}^i) + \nabla f_i (x_{t,k-1}^i) \nonumber\\
     & + \nabla f_s (x_t)  - g_t^i + \nabla f_i (x_t) - \nabla f_i (x_t)) - x_{t}\|^2] \nonumber\\
     & \leq (1 + \frac{1}{2K-1}) \mathbb{E}_t [\|x_{t,k-1}^i - x_t\|^2] +  2\eta^2\sigma_l^2\nonumber\\ 
     & +  2K \eta ^2 \mathbb{E}_t [\|\nabla f_i (x_{t,k-1}^i) + \nabla f_s (x_t)  - \nabla f_i (x_t)\|^2] \nonumber\\
     & = (1 + \frac{1}{2K-1}) \mathbb{E}_t [\|x_{t,k-1}^i - x_t\|^2] +  2\eta^2\sigma_l^2\nonumber\\ 
     & +  2K \eta ^2 \mathbb{E}_t [\|\nabla f_i (x_{t,k-1}^i) + \nabla f_s (x_t)  - \nabla f_i (x_t)  - \nabla f (x_t) + \nabla f (x_t)\|^2] \nonumber\\
     & \leq (1 + \frac{1}{2K-1}) \mathbb{E}_t [\|x_{t,k-1}^i - x_t\|^2] +  6K \eta ^2 L^2 \mathbb{E}_t [\|x_{t,k-1}^i - x_t\|^2] \nonumber\\ 
     & +  6K \eta ^2 \mathbb{E}_t [\|\nabla f_s (x_t) - \nabla f (x_t)\|^2]  +  6K \eta ^2 \mathbb{E}_t [\|\nabla f (x_t)\|^2] + 2\eta^2\sigma_l^2\nonumber\\
     & \leq (1 + \frac{1}{2K-1} + 6K \eta ^2 L^2) \mathbb{E}_t [\|x_{t,k-1}^i - x_t\|^2] \nonumber\\
     & + \frac{6K \eta ^2 \sigma^2}{m_s} +  6K \eta ^2 \mathbb{E}_t [\|\nabla f (x_t)\|^2]+ 2\eta^2\sigma_l^2\nonumber\\
     & \leq (1 + \frac{1}{K-1}) \mathbb{E}_t [\|x_{t,k-1}^i - x_t\|^2] + \frac{6K \eta ^2 \sigma^2}{m_s}\nonumber\\
     & +  6K \eta ^2 \mathbb{E}_t [\|\nabla f (x_t)\|^2] + 2\eta^2\sigma_l^2
\end{align}
We then average over each client $i$, and then have:
\begin{align}
     &\frac{1}{N} \sum_{i\in [N]}\mathbb{E}_t [\| x_{t,k}^i - x_{t}\|^2]  \leq (1 + \frac{1}{K-1}) \frac{1}{N} \sum_{i\in [N]}\mathbb{E}_t [\|x_{t,k-1}^i - x_t\|^2] \nonumber\\
     & + \frac{6K \eta ^2 \sigma^2}{m_s} +  6K \eta ^2 \mathbb{E}_t [\|\nabla f (x_t)\|^2] + 2\eta^2\sigma_l^2\nonumber\\
     & \leq (K-1) [(1 + \frac{1}{K-1})^K - 1][ \frac{6K \eta ^2 \sigma^2}{m_s} \nonumber\\
     & + 2\eta^2\sigma_l^2 + 6K \eta ^2 \mathbb{E}_t [\|\nabla f (x_t)\|^2]]\nonumber\\
     &\leq \frac{30K^2 \eta ^2 \sigma^2}{m_s} +  30K^2 \eta ^2 \mathbb{E}_t [\|\nabla f (x_t)\|^2] + 10K \eta^2\sigma_l^2
\end{align}
Substituting to $C_5$, we have:
\begin{align}
    &C_5 \leq \frac{\eta\eta_g L^2}{2N}\sum_{i\in [N]}\sum_{k=0}^{K-1}\mathbb{E}_t [\| x_{t,k}^i - x_{t}\|^2]\nonumber\\
    & \leq \frac{\eta\eta_g L^2}{2}\sum_{k=0}^{K-1}( \frac{30K^2 \eta ^2 \sigma^2}{m_s} +  30K^2 \eta ^2 \mathbb{E}_t [\|\nabla f (x_t)\|^2] + 10K \eta^2\sigma_l^2) \nonumber\\
    & \leq  \frac{15K^3L^2 \eta ^3 \eta_g \sigma^2}{m_s} +  \frac{\eta \eta_g K}{4} \mathbb{E}_t [\|\nabla f (x_t)\|^2] + 5K^2L^2\eta^3\eta_g\sigma_l^2
\end{align}
}

\textcolor{black}{
Then the term $C_3$ can be bounded as:
\begin{align}
    &C_3 \leq \frac{15K^3L^2 \eta ^3 \eta_g \sigma^2}{m_s}  - \frac{\eta\eta_g K}{4}\|\nabla f(x_{t})\|^2 + 5K^2L^2\eta^3\eta_g\sigma_l^2\nonumber\\
    &- \frac{\eta\eta_g}{2K}\mathbb{E}_t [\|\frac{1}{N} \sum_{i\in [N]}\sum_{k=0}^{K-1} \nabla f_i(x_{t,k}^i)\|^2].
\end{align}
}
\textcolor{black}{
The remaining term in $C_1$ is $C_4$ which can be bounded as:
\begin{align}
    & C_4 = \underbrace{\frac{\gamma}{2E} \mathbb{E}_t[\|\sum_{e=0}^{E-1} \nabla f(x_{t+1,e}^s) - E\nabla f(x_t)\|^2]}_{C_7} \nonumber\\
    &- \frac{\gamma E}{2}\|\nabla f(x_{t})\|^2 - \frac{\gamma}{2E}\mathbb{E}_t[\|\sum_{e=0}^{E-1} \nabla f(x_{t+1,e}^s)\|^2]. \label{ineq_c4}
\end{align}
Then the term $C_7$ is bounded by:
\begin{align}
    & C_7 \leq \underbrace{\frac{\gamma L^2}{2}\sum_{e=0}^{E-1} \mathbb{E}_t[\| x_{t+1,e}^s - x_t\|^2]}_{C_8} \nonumber\\
    & \leq \underbrace{{\gamma L^2}\sum_{e=0}^{E-1} \mathbb{E}_t[\| x_{t+1,e}^s - x_{t+1}^s\|^2]}_{C_{9}} + \eta^2 \eta_g^2\gamma E L^2\nonumber\\
    &  \underbrace{\mathbb{E}_t [\|\frac{1}{M} \sum_{i\in \mathcal{S}_t}\sum_{k=0}^{K-1} (\nabla f_i(x_{t,k}^i)+\nabla f_s(x_t)-\nabla f_i(x_t)) \|^2]}_{C_{10}}\nonumber\\
    & +  \frac{2\eta^2 \eta_g^2\gamma E L^2K \sigma_l^2}{M}
\end{align}
Then bounding $C_{9}$ which is same as $T_{10}$, we have:
\begin{align}
    & C_{9} \leq \frac{3\gamma^3 E^2 L^2 \sigma^2}{m_s} + 3\gamma^3E^2 L^4\sum_{e=0}^{E-1}\mathbb{E}_t[\|x_{t+1, e}^s -  x_{t}\|^2] \nonumber\\
    & + 3\gamma^3E^2 L^2\sum_{e=0}^{E-1}\mathbb{E}_t[\| \nabla f(x_{t})\|^2],
\end{align}
Then to bound the term $C_{10}$, we have:
\begin{align}
    & C_{10} \leq 3\mathbb{E}_t [ \|\frac{1}{M} \sum_{i\in \mathcal{S}_t}\sum_{k=0}^{K-1}[ \nabla f_i(x_{t,k}^i) - \nabla f_i(x_{t}) ]\|^2]\nonumber\\
    & + 3\mathbb{E}_t [ \|\frac{1}{M} \sum_{i\in \mathcal{S}_t}\sum_{k=0}^{K-1}[\nabla f_s(x_{t}) -  \nabla f(x_{t}) ]\|^2] \nonumber\\
    & + 3\mathbb{E}_t [ \|\frac{1}{M} \sum_{i\in \mathcal{S}_t}\sum_{k=0}^{K-1} \nabla f(x_{t}) \|^2] \nonumber\\
    & \leq \frac{3}{N} \sum_{i\in [N]}\mathbb{E}_t [\|\sum_{k=0}^{K-1}[ \nabla f_i(x_{t,k}^i) - \nabla f_i(x_{t}) ]\|^2]\nonumber\\
    & + 3\mathbb{E}_t [ \|\frac{1}{M} \sum_{i\in \mathcal{S}_t}\sum_{k=0}^{K-1}[\nabla f_s(x_{t}) -  \nabla f(x_{t}) ]\|^2] \nonumber\\
    & + 3 K^2\mathbb{E}_t [ \|\nabla f(x_{t}) \|^2] \nonumber\\
    & = \frac{9K^2}{2}\mathbb{E}_t [\|\nabla f(x_{t})\|^2] +  \frac{90K^4L^2 \eta ^2 \sigma^2}{m_s}\nonumber\\
    & + 30 K^3 L^2 \eta^2 \sigma_l^2 + \frac{3K\sigma^2}{M m_s},
\end{align}
where the second term in the second equality is due to the server's uniformly selection without replacement and the fact $\mathbb{E} [\|x_1 + \cdots + x_n \|^2] = \mathbb{E}[\|x_1\|^2 + \cdots + \| x_n \|^2]$ if $x_i$ is independent with zero mean.
Substituting the results of $C_{9}$ and $C_{10}$, we have:
\begin{align}
    & C_8 \leq \frac{3\gamma^3 E^2 L^2 \sigma^2}{m_s} + 3\gamma^3E^2 L^4\sum_{e=0}^{E-1}\mathbb{E}_t[\|x_{t+1, e}^s -  x_{t}\|^2] \nonumber\\
    & + 3\gamma^3E^3 L^2\mathbb{E}_t[\| \nabla f(x_{t})\|^2] + \frac{9}{2}\eta^2 \eta_g^2\gamma E L^2 K^2\mathbb{E}_t [\|\nabla f(x_{t})\|^2] \nonumber\\
    & + \frac{90\eta^4 \eta_g^2\gamma E L^4 K^4\sigma^2}{m_s} + \frac{3\eta^2 \eta_g^2\gamma E L^2 K\sigma^2}{M m_s} \nonumber\\
    & + 5 \eta^4 \eta_g^2 L^3 K^3 \sigma_l^2 + \frac{\eta^2 \eta_g^2 LK \sigma_l^2}{3M} \nonumber\\
    & \leq \frac{3\gamma^3 E^2 L^2 \sigma^2}{(1-\mathcal{A})m_s} + \frac{\mathcal{A}\gamma E}{2(1 - \mathcal{A})}\mathbb{E}_t[\| \nabla f(x_{t})\|^2] \nonumber\\
    & + \frac{3\eta^2 \eta_g^2 L K^2}{4(1 - \mathcal{A})}\mathbb{E}_t [\|\nabla f(x_{t})\|^2] + \frac{15\eta^4 \eta_g^2 L^3 K^4\sigma^2}{m_s(1 - \mathcal{A})} + \frac{\eta^2 \eta_g^2 L K\sigma^2}{2M m_s(1 - \mathcal{A})} \nonumber\\
    & + \frac{5 \eta^4 \eta_g^2 L^3 K^3 \sigma_l^2}{1-\mathcal{A}} + \frac{\eta^2 \eta_g^2 LK \sigma_l^2}{3M(1-\mathcal{A})} \nonumber\\
    & \leq \frac{18\gamma^3 E^2 L^2 \sigma^2}{5m_s} + \frac{\gamma E}{10}\mathbb{E}_t[\| \nabla f(x_{t})\|^2] + \frac{3\eta^2 \eta_g^2 L K\sigma^2}{5M m_s}\nonumber\\
    & + \frac{9\eta^2 \eta_g^2 L K^2}{10}\mathbb{E}_t [\|\nabla f(x_{t})\|^2] + \frac{18\eta^4 \eta_g^2 L^3 K^4\sigma^2}{m_s} \nonumber\\
    & + 6 \eta^4 \eta_g^2 L^3 K^3 \sigma_l^2 + \frac{2\eta^2 \eta_g^2 LK \sigma_l^2}{5M} 
\end{align}
where $\mathcal{A} = 6\gamma^2L^2E^2$ and let $\gamma \leq \frac{1}{6LE}$ such that $\mathcal{A} \leq \frac{1}{6}$, $\frac{1}{1 - \mathcal{A}} \leq \frac{6}{5}$ and $\frac{\mathcal{A}}{1 - \mathcal{A}} \leq \frac{1}{5}$. Substituting the above result to $C_4$, we have:
\begin{align}
    &C_4 \leq \frac{18\gamma^3 E^2 L^2 \sigma^2}{5m_s} - \frac{2\gamma E}{5}\mathbb{E}_t[\| \nabla f(x_{t})\|^2] + \frac{3\eta^2 \eta_g^2 L K\sigma^2}{5M m_s}\nonumber\\
    & + \frac{9\eta^2 \eta_g^2 L K^2}{10}\mathbb{E}_t [\|\nabla f(x_{t})\|^2] + \frac{18\eta^4 \eta_g^2 L^3 K^4\sigma^2}{m_s}\nonumber\\
     & + 6 \eta^4 \eta_g^2 L^3 K^3 \sigma_l^2 + \frac{2\eta^2 \eta_g^2 LK \sigma_l^2}{5M} - \frac{\gamma}{2E}\mathbb{E}_t[\|\sum_{e=0}^{E-1} \nabla f(x_{t+1,e}^s)\|^2].
\end{align}
Combing the inequalities of $C_3$ and $C_4$, we can then bound $C_1$ as:
\begin{align}
    & C_1 \leq  \frac{15K^3L^2 \eta ^3 \eta_g \sigma^2}{m_s} - (\frac{\eta\eta_g K}{4} + \frac{2\gamma E}{5})\|\nabla f(x_{t})\|^2 \nonumber\\
    &- \frac{\eta\eta_g}{2K}\mathbb{E}_t [\|\frac{1}{N} \sum_{i\in [N]}\sum_{k=0}^{K-1} \nabla f_i(x_{t,k}^i)\|^2] + \frac{18\gamma^3 E^2 L^2 \sigma^2}{5m_s} \nonumber\\
    & + \frac{3\eta^2 \eta_g^2 L K\sigma^2}{5M m_s} + \frac{9\eta^2 \eta_g^2 L K^2}{10}\mathbb{E}_t [\|\nabla f(x_{t})\|^2] \nonumber\\
    & + \frac{18\eta^4 \eta_g^2 L^3 K^4\sigma^2}{m_s}  - \frac{\gamma}{2E}\mathbb{E}_t[\|\sum_{e=0}^{E-1} \nabla f(x_{t+1,e}^s)\|^2]\nonumber\\
     & + 6 \eta^4 \eta_g^2 L^3 K^3 \sigma_l^2 + \frac{2\eta^2 \eta_g^2 LK \sigma_l^2}{5M} + 5K^2L^2\eta^3\eta_g \sigma_l^2.
\end{align}
}

\textcolor{black}{
Then we need to bound the term $C_2$ in \ref{ineq_smooth_c}:
\begin{align}
    & C_2 = \frac{L}{2} \mathbb{E}_t [\| \eta\eta_g \frac{1}{M} \sum_{i\in \mathcal{S}_t}\sum_{k=0}^{K-1} (g_{t,k}^i+\nabla f_s(x_t)-g_t^i) \nonumber\\
    & + \gamma\sum_{e=0}^{E-1} \nabla f_s(x_{t+1,e}^s) \|^2] \nonumber\\
    & \leq L\mathbb{E}_t [\| \eta\eta_g \frac{1}{M} \sum_{i\in \mathcal{S}_t}\sum_{k=0}^{K-1} (\nabla f_i(x_{t,k}^i)+\nabla f_s(x_t)-\nabla f_i(x_t)) \|^2] \nonumber\\
    & + \frac{2\eta^2\eta_g^2LK\sigma_l^2}{M}+ L\mathbb{E}_t [\|\gamma\sum_{e=0}^{E-1} \nabla f_s(x_{t+1,e}^s) \|^2]  \nonumber\\
    & \leq \eta^2\eta_g^2L\underbrace{\mathbb{E}_t [\|\frac{1}{M} \sum_{i\in \mathcal{S}_t}\sum_{k=0}^{K-1} (\nabla f_i(x_{t,k}^i)+\nabla f_s(x_t)-\nabla f_i(x_t)) \|^2] }_{C_{10}}\nonumber\\
    & + \gamma^2L\underbrace{\mathbb{E}_t [\|\sum_{e=0}^{E-1} \nabla f_s(x_{t+1,e}^s) \|^2]}_{C_{11}}  + \frac{2\eta^2\eta_g^2LK\sigma_l^2}{M}. \label{ineq_t2_c}
\end{align}
The bound of $C_{11}$ is exactly the same as the bound of $T_{13}$. Thus, we have:
\begin{align}
    &C_{11} \leq \frac{ E \sigma^2}{m_s} + \mathbb{E}_t [\|\sum_{e=0}^{E-1} \nabla f(x_{t+1,e}^s)\|^2].
\end{align}
Substituting the result of $C_{10}$ and $C_{11}$, we can finally bound $C_2$ as:
\begin{align}
    &C_2 \leq \frac{ 90\eta^4\eta_g^2 L^3 K^4\sigma^2}{m_s} + \frac{9}{2}K^2\eta^2\eta_g^2L\mathbb{E}_t [\|\nabla f(x_{t})\|^2] \nonumber\\
    & + \frac{3\eta^2\eta_g^2KL\sigma^2}{M m_s} + \frac{ \gamma^2L E \sigma^2}{m_s} + \gamma^2L\mathbb{E}_t [\|\sum_{e=0}^{E-1} \nabla f(x_{t+1,e}^s)\|^2] \nonumber\\
    &  + \frac{2\eta^2\eta_g^2LK\sigma_l^2}{M} + 30 K^3 L^3\eta^4\eta_g^2\sigma_l^2.
\end{align}
With both $C_1$ and $C_2$ bounded, we finally have:
\begin{align}
    & \mathbb{E}_t [f(x_{t+1})] \leq f(x_t) - (\frac{\eta\eta_g K}{10} + \frac{2\gamma E}{5})\|\nabla f(x_{t})\|^2 \nonumber\\
    & + \frac{18\eta^3\eta_g L^2 K^3\sigma^2}{m_s} + \frac{8\gamma^2 E L \sigma^2}{5m_s} + \frac{18\eta^2 \eta_g^2 L K\sigma^2}{5M m_s} \nonumber\\
    & + \frac{12\eta^2\eta_g^2LK\sigma_l^2}{5M} + 6\eta^3 \eta_g L^2 K^2 \sigma_L^2,
\end{align}
where $\gamma \leq \frac{1}{6EL}$, $\eta \leq \frac{1}{6KL}$ and $\eta \eta_g \leq \frac{1}{36KL}$.
Rearranging and summing from $t = 0, \dots, T-1$, we have the convergence as:
\begin{align}
    & \min_{t \in [T]} \mathbb{E}\|\nabla f(x_t)\|_2^2 = \mathcal{O}\bigg(\frac{(f_0 - f_*)}{T(\gamma E + \eta\eta_g K)}\bigg)\ \nonumber\\
    & + \mathcal{O}\bigg(\frac{\eta^3\eta_g L^2 K^3\sigma^2}{m_s(\gamma E + \eta\eta_g K)}\bigg) + \mathcal{O}\bigg(\frac{\gamma^2L E \sigma^2}{m_s(\gamma E + \eta\eta_g K)} \bigg) \nonumber\\
    & + \mathcal{O}\bigg(\frac{\eta^2\eta_g^2KL\sigma^2}{M m_s(\gamma E + \eta\eta_g K)} \bigg)  + \mathcal{O}\bigg(\frac{\eta^3\eta_g L^2 K^2\sigma_l^2}{\gamma E + \eta\eta_g K}\bigg) \nonumber\\
    & +\mathcal{O}\bigg(\frac{\eta^2\eta_g^2LK \sigma_l^2}{M(\gamma E + \eta\eta_g K)}\bigg) ,
\end{align}
where $f_0 = f(x_0)$, $f_* = f(x_*)$.}

\section{Proof of Theorem 3}
\label{proof-t3}
\textcolor{black}{
After each global round $t$, we have the new model $x_{t+1}$ as:
\begin{align}
    & x_{t+1}  =  x_{t} - \eta\eta_g \frac{1}{M} \sum_{i\in \mathcal{S}_t}\sum_{k=0}^{K-1}g_{t,k}^i \nonumber\\
    & - \frac{K\eta\eta_g}{M}\sum_{i\in \mathcal{S}_t}( \nabla f_s(x_t) - g_t^i) - \gamma\sum_{e=0}^{E-1} \nabla f_s(x_{t+1,e}^s).
\end{align}
}
Due to the smoothness in Assumption 1, taking expectation of $f(x_{t+1})$ over the randomness at communication round $t$, we have:
\begin{align}
    \mathbb{E}_t [f(x_{t+1})] &\leq f(x_t) + \underbrace{\bigg< \nabla f(x_t), \mathbb{E}_t [x_{t+1} - x_t] \bigg>}_{S_1}\nonumber\\
    & + \underbrace{\frac{L}{2} \mathbb{E}_t [\| x_{t+1} - x_t \|^2]}_{S_2}. \label{ineq_smooth_s}
\end{align}
\textcolor{black}{
Similar to the proof of theorem 2, focusing on $\mathbb{E}_t [x_{t+1} - x_t]$ in the term $S_1$, we can find that:
\begin{align}
    & \mathbb{E}_t [x_{t+1} - x_t]\nonumber\\
    & = -\mathbb{E}_t [\eta\eta_g \frac{1}{M} \sum_{i\in \mathcal{S}_t}\sum_{k=0}^{K-1} \nabla f_i(x_{t,k}^i) + \gamma\sum_{e=0}^{E-1} \nabla f_s(x_{t+1,e}^s)] \nonumber\\
    & - \underbrace{\mathbb{E}_t[\frac{K\eta\eta_g}{M}\sum_{i\in \mathcal{S}_t}( \nabla f_s(x_t) - g_{t}^i)]}_{=0}
\end{align}
}
The term $S_1$ can be written as:
\begin{align}
    & S_1 = \underbrace{-\eta\eta_g\bigg< \nabla f(x_t), \mathbb{E}_t [\frac{1}{M} \sum_{i\in \mathcal{S}_t}\sum_{k=0}^{K-1} \nabla f_i(x_{t,k}^i)]\bigg>}_{S_3}\nonumber\\
    & \underbrace{- \gamma\bigg<\nabla f(x_t), \mathbb{E}_t[\sum_{e=0}^{E-1} \nabla f_s(x_{t+1,e}^s)] \bigg>}_{S_4}.
\end{align}

\textcolor{black}{
The term $S_3$ can be bounded as follows:
\begin{align}
    & S_3 = \underbrace{\frac{\eta\eta_g}{2K}\mathbb{E}_t [\|\frac{1}{N} \sum_{i\in [N]}\sum_{k=0}^{K-1} \nabla f_i(x_{t,k}^i) - K\nabla f(x_{t})\|^2]}_{S_5} \nonumber\\
     & - \frac{\eta\eta_g K}{2}\|\nabla f(x_{t})\|^2 - \frac{\eta\eta_g}{2K}\mathbb{E}_t [\|\frac{1}{N} \sum_{i\in [N]}\sum_{k=0}^{K-1} \nabla f_i(x_{t,k}^i)\|^2]. \label{ineq_s3}
\end{align}
Then we can bound $S_5$ as:
\begin{align}
    &S_5  \leq \underbrace{\frac{\eta\eta_g L^2}{2N}\sum_{i\in [N]}\sum_{k=0}^{K-1} \mathbb{E}_t [\| x_{t,k}^i - x_{t}\|^2] }_{S_6}
\end{align}
In FedCLG-S, the correction step occurs during server aggregation (from $x_t$ to $x_{t+1}^s$). Therefore, the difference between the local model $x_{t,k}^i$ and the initial model $x_t$ is the same as $T_6$, which was previously defined in the proof of Theorem 1. 
\begin{align}
    & S_6  \leq \frac{\eta^3\eta_g L^2 K^2}{2} \sum_{k=0}^{K-1} \underbrace{\frac{1}{N}\sum_{i\in [N]}\mathbb{E}_t [\|\nabla f_i(x_{t, k}^i)\|^2]}_{S_7} \nonumber\\
    & + \frac{\eta^3\eta_g L^2 K^2 \sigma_l^2}{2}.
    \label{ineq_s6}
\end{align}
To further bound $S_7$ which is corresponding to $T_7$, we have:
\begin{align}
    & S_7 \leq \frac{3 L^2}{N} \sum_{i\in [N]} \mathbb{E}_t [\| x_{t} - x_{t, k}^i\|^2] + 3 \sigma_g^2 + 3\mathbb{E}_t [\|\nabla f(x_{t})\|^2], 
\end{align}
Substituting the result of $S_7$ to the inequality of $S_6$, we have:
\begin{align}
    & S_6 \leq \frac{3\eta^3\eta_g L^2 K^2 (\sigma_l^2+3K\sigma_g^2)}{4} + \frac{\eta\eta_g K}{4}\mathbb{E}_t [\|\nabla f(x_{t})\|^2].
\end{align}
Then we have $S_3$ be bounded as:
\begin{align}
    &S_3\leq \frac{3\eta^3\eta_g L^2 K^2 (\sigma_l^2+3K\sigma_g^2)}{4}  - \frac{\eta\eta_g K}{4}\|\nabla f(x_{t})\|^2 \nonumber\\
    &- \frac{\eta\eta_g}{2K}\mathbb{E}_t [\|\frac{1}{N} \sum_{i\in [N]}\sum_{k=0}^{K-1} \nabla f_i(x_{t,k}^i)\|^2].
\end{align}
}
The remaining term in $S_1$ is $S_4$ can be bounded as:
\begin{align}
    & S_4  = \underbrace{\frac{\gamma}{2E} \mathbb{E}_t[\|\sum_{e=0}^{E-1} \nabla f(x_{t+1,e}^s) - E\nabla f(x_t)\|^2]}_{S_8} \nonumber\\
    &- \frac{\gamma E}{2}\|\nabla f(x_{t})\|^2 - \frac{\gamma}{2E}\mathbb{E}_t[\|\sum_{e=0}^{E-1} \nabla f(x_{t+1,e}^s)\|^2]. \label{ineq_s4}
\end{align}
\textcolor{black}{
Then the term $S_8$ is bounded as:
\begin{align}
    & S_8 \leq \underbrace{\frac{\gamma L^2}{2}\sum_{e=0}^{E-1} \mathbb{E}_t[\| x_{t+1,e}^s - x_t\|^2]}_{S_9} \nonumber\\
    & \leq \underbrace{{\gamma L^2}\sum_{e=0}^{E-1} \mathbb{E}_t[\| x_{t+1,e}^s - x_{t+1}^s\|^2]}_{S_{10}}  + \frac{2\eta^2 \eta_g^2\gamma E L^2 K \sigma_l^2}{M}\nonumber\\ 
    & +\eta^2 \eta_g^2\gamma E L^2 \nonumber\\
    & \underbrace{\mathbb{E}_t [\|\frac{1}{M} \sum_{i\in \mathcal{S}_t}\sum_{k=0}^{K-1}\nabla f_i(x_{t,k}^i)+K(\nabla f_s(x_t)-\nabla f_i(x_t)) \|^2]}_{S_{11}}\nonumber\\
\end{align}
}
\textcolor{black}{
The bounds of $S_{10}$ and $S_{11}$ are:
\begin{align}
    & S_{10} \leq \frac{3\gamma^3 E^2 L^2 \sigma^2}{m_s} + 3\gamma^3E^2 L^4\sum_{e=0}^{E-1}\mathbb{E}_t[\|x_{t+1, e}^s -  x_{t}\|^2] \nonumber\\
    & + 3\gamma^3E^2 L^2\sum_{e=0}^{E-1}\mathbb{E}_t[\| \nabla f(x_{t})\|^2],
\end{align}
and 
\begin{align}
    & S_{11} \leq \frac{9K^2}{2}\mathbb{E}_t [\|\nabla f(x_{t})\|^2] +  \frac{9\eta^2 L^2 K^3 (\sigma_l^2+3K\sigma_g^2) }{2} + \frac{3K\sigma^2}{M m_s},
\end{align}}
\textcolor{black}{
Substituting the results of $S_{10}$ and $S_{11}$, we have:
\begin{align}
    &S_{9}  \leq  \frac{18\gamma^3 E^2 L^2 \sigma^2}{5m_s} + \frac{\gamma E}{10} \mathbb{E}_t[\| \nabla f(x_{t})\|^2] \nonumber\\
    & +\frac{9\eta^4\eta_g^2 L^3 K^3(\sigma_l^2+3K\sigma_g^2)}{10} + \frac{9\eta^2 \eta_g^2LK^2}{10}\mathbb{E}_t [\|\nabla f(x_{t})\|^2]  \nonumber\\
    &  + \frac{3\eta^2 \eta_g^2 L K\sigma^2}{5M m_s} + \frac{2\eta^2 \eta_g^2 L K \sigma_l^2}{5M},
\end{align}
Then we have $S_4$ as:
\begin{align}
    & S_4 \leq  \frac{18\gamma^3 E^2 L^2 \sigma^2}{5m_s} - \frac{2\gamma E}{5} \mathbb{E}_t[\| \nabla f(x_{t})\|^2] \nonumber\\
    & + \frac{9\eta^4\eta_g^2 L^3 K^3(\sigma_l^2+3K\sigma_g^2)}{10}  + \frac{9\eta^2 \eta_g^2LK^2}{10}\mathbb{E}_t [\|\nabla f(x_{t})\|^2] \nonumber\\
    & +  \frac{3\eta^2 \eta_g^2 L K\sigma^2}{5M m_s}  - \frac{\gamma}{2E}\mathbb{E}_t[\|\sum_{e=0}^{E-1} \nabla f(x_{t+1,e}^s)\|^2]  + \frac{2\eta^2 \eta_g^2 L K \sigma_l^2}{5M}.
\end{align}
Combing the inequalities of $S_3$ and $S_4$, we can then bound $S_1$ as:
\begin{align}
    &S_1 \leq \frac{3\eta^3\eta_g L^2 K^2 (\sigma_l^2+3K\sigma_g^2)}{4} - (\frac{2\gamma E}{5}+\frac{\eta\eta_g K}{4})\|\nabla f(x_{t})\|^2 \nonumber\\
    & - \frac{\eta\eta_g}{2K}\mathbb{E}_t [\|\frac{1}{N} \sum_{i\in [N]}\sum_{k=0}^{K-1} \nabla f_i(x_{t,k}^i)\|^2]  + \frac{18\gamma^3 E^2 L^2 \sigma^2}{5m_s}\nonumber\\
    & + \frac{9\eta^4\eta_g^2 L^3 K^3(\sigma_l^2+3K\sigma_g^2)}{10}   + \frac{9\eta^2 \eta_g^2LK^2}{10}\mathbb{E}_t [\|\nabla f(x_{t})\|^2] \nonumber\\
    & +  \frac{3\eta^2 \eta_g^2 L K\sigma^2}{5M m_s} - \frac{\gamma}{2E}\mathbb{E}_t[\|\sum_{e=0}^{E-1} \nabla f(x_{t+1,e}^s)\|^2] + \frac{2\eta^2 \eta_g^2 L K \sigma_l^2}{5M}.
\end{align}
}
\textcolor{black}{
Utilizing the similar steps in the proof of Thm1 and Thm2, we can bound $S_2$ as:
\begin{align}
    &S_2 \leq \frac{9\eta^4 \eta_g^2 L^3 K^3(3K\sigma_g^2 + \sigma_l^2) }{2} + \frac{9K^2\eta^2\eta_g^2L}{2}\mathbb{E}_t [\|\nabla f(x_{t})\|^2] \nonumber\\
    & + \frac{3\eta^2\eta_g^2KL\sigma^2}{M m_s} + \frac{ \gamma^2L E \sigma^2}{m_s} + \gamma^2L\mathbb{E}_t [\|\sum_{e=0}^{E-1} \nabla f(x_{t+1,e}^s)\|^2] \nonumber\\
    &  + \frac{2\eta^2\eta_g^2LK\sigma_l^2}{M}.
\end{align}
Finally, we can get the convergence bound for FedCLG-S as:
\begin{align}
    & \min_{t \in [T]} \mathbb{E}\|\nabla f(x_t)\|_2^2 = \mathcal{O}\bigg(\frac{(f_0 - f_*)}{T(\gamma E + \eta\eta_g K)}\bigg)\ \nonumber\\
    & + \mathcal{O}\bigg(\frac{\eta^3\eta_g L^2 K^3\sigma_g^2}{\gamma E + \eta\eta_g K}\bigg) + \mathcal{O}\bigg(\frac{\gamma^2L E \sigma^2}{m_s(\gamma E + \eta\eta_g K)} \bigg) \nonumber\\
    & + \mathcal{O}\bigg(\frac{\eta^2\eta_g^2KL\sigma^2}{M m_s(\gamma E + \eta\eta_g K)} \bigg) + \mathcal{O}\bigg(\frac{\eta^3\eta_g L^2 K^2\sigma_l^2}{\gamma E + \eta\eta_g K}\bigg) \nonumber\\
    & +\mathcal{O}\bigg(\frac{\eta^2\eta_g^2LK \sigma_l^2}{M(\gamma E + \eta\eta_g K)}\bigg) ,
\end{align}
where $f_0 = f(x_0)$, $f_* = f(x_*)$ and  $\gamma \leq \frac{1}{6EL}$, $\eta \leq \frac{1}{3KL}$ and $\eta \eta_g \leq \frac{1}{27KL}$. 
}

\end{document}